\documentclass[submission,copyright,creativecommons]{eptcs}
\usepackage[latin1]{inputenc}
\providecommand{\event}{LSFA 2021} 
\usepackage{breakurl}             
\usepackage{underscore}           
\usepackage{amssymb}
\usepackage{graphicx}
\usepackage{proof}
\usepackage[dvipsnames]{xcolor}
\usepackage{amsthm}
\usepackage{subfigure}
\usepackage{hyperref}
\usepackage{stmaryrd}
\usepackage{latexsym}
\usepackage{amsmath}
\usepackage{amssymb}
\usepackage{mdwlist}
\usepackage{graphicx}
\usepackage{proof}
\usepackage{xspace}
\usepackage{mycommands}
\usepackage{balance}
\usepackage{comment}
\usepackage{color}
\usepackage{url}
\usepackage{ifpdf}
\usepackage{tikz}
\usepackage{todonotes}

\usepackage[normalem]{ulem}

\usepackage{url}
\urldef{\mailsa}\path|{alfred.hofmann, ursula.barth, ingrid.haas, frank.holzwarth,|
\urldef{\mailsb}\path|anna.kramer, leonie.kunz, christine.reiss, nicole.sator,|
\urldef{\mailsc}\path|erika.siebert-cole, peter.strasser, lncs}@springer.com|

\title{A Subexponential View of Domains in Session Types}
\author{Daniele Nantes
\institute{MAT-UnB\\ Brasilia, Brazil}
\email{dnantes@unb.br}
\and
Carlos Olarte
\institute{ECT-UFRN\\Natal, Brazil\\ LIPN, Universit\'e Sorbonne Paris Nord\\Villetaneuse, France
}
\email{olarte@lipn.univ-paris13.fr}
\and
Daniel Ventura
\institute{INF-UFG\\ Goi\^ania, Brazil}
\email{ventura@ufg.br}
}
\def\titlerunning{A Subexponential View of Domains in Session Types}
\def\authorrunning{Nantes, Olarte and Ventura}
\begin{document}
\maketitle

\begin{abstract}

Linear logic (LL) has inspired the design of many computational systems, offering  reasoning techniques  built on top of its meta-theory. Since its
inception, several connections between concurrent systems and LL have emerged
from different perspectives.  In the last decade, the seminal work of Caires
and Pfenning showed that formulas in LL can be interpreted as session types and
processes in the $\pi$-calculus as proof terms.  This leads to a Curry-Howard
interpretation where  proof reductions in the cut-elimination
procedure correspond to process reductions/interactions.  The subexponentials
in LL have also played an important role in concurrent systems since they can
be interpreted in different ways, including timed, spatial and even epistemic
modalities in distributed systems. In this paper we address the question: What
is the meaning  of the subexponentials from the point of view of a session type
interpretation? Our answer is a $\pi$-like process calculus  where agents
reside in locations/sites and they make it explicit how the communication among
the different sites should happen. The design of this language relies
completely on the proof theory of the subexponentials in LL,  thus extending
the Caires-Pfenning interpretation in an elegant way. 

\end{abstract}

\section{Introduction}\label{sec:intro}

One of the most profound connections 
between proof theory  and computation is the  Curry-Howard isomorphism (also known as \emph{formulas-as-types}) where (functional) programs correspond to  proofs in intuitionistic logic, and reductions/computations in one system correspond
to the process of normalization in the other \cite{Sorensen06}. Several correspondences with formal systems in a Curry-Howard fashion has been considered ever since, including classical logic {\cite{DBLP:conf/popl/Griffin90,DBLP:conf/lpar/Parigot92,DBLP:journals/lmcs/KesnerV19}}, intuitionistic modal logic {\cite{DBLP:journals/jacm/DaviesP01,DBLP:journals/tocl/NanevskiPP08}} and intuitionistic linear logic {\cite{DBLP:journals/tcs/Abramsky93,DBLP:conf/tlca/BentonBPH93}}, to name a few. 

Inspired by the \emph{proofs-as-processes} paradigm {\cite{DBLP:journals/tcs/Abramsky94}}, session types were proposed as a type foundation of communicating processes {\cite{DBLP:conf/esop/HondaVK98}} with a wide range of applications {\cite{DBLP:conf/wsfm/Dezani-Ciancaglinid09,DBLP:journals/csur/HuttelLVCCDMPRT16}}.
The seminal work of Caires and Pfenning \cite{DBLP:conf/concur/CairesP10,DBLP:journals/mscs/CairesPT16}
establishes a  correspondence between Girard's intuitionistic  linear logic \cite{girard87tcs} and   the session-typed $\pi$-calculus, a process calculus with mobility \cite{DBLP:journals/iandc/MilnerPW92a}. This correspondence was  recently
extended in  \cite{DBLP:conf/concur/CairesPPT19}, dealing with a notion of domain/location not featured in  session type systems. 
 The present work proposes an alternative interpretation 
  to obtain a domain-aware session type system. For that, we build on a different logical system and, we design a language of processes  with explicit annotations  for locations.
  Let us elaborate more on these two distinguished characteristics of our approach.

\noindent\textbf{Linear logic and variants.} Linear logic  (LL) \cite{girard87tcs} is a substructural  logic that can be seen as a refinement of classical and intuitionistic logics, integrating the dualities of the former with many of the constructive properties of the latter.
In this logic, formulas are seen as resources that are consumed when used. For instance, the linear implication 
$A \limp B$ specifies that, in order to produce $B$, one copy of $A$ must be consumed. In order to specify that a formula $A$ can be used unboundedly  many times, as in classical logic, the formula $A$ must be marked with the connective $!$, called exponential (or bang). 

Apart from substructural behaviors (being linear or unbounded), it is not
natural to specify some other modal behaviors directly   in LL  \cite{DBLP:conf/lpar/LellmannOP17}. For this reason,  recent developments in  LL have brought new systems  to specify multimodalities in it. 
One of such systems
is Hybrid Linear Logic (HyLL) \cite{DBLP:conf/types/DespeyrouxC13,DBLP:journals/mscs/ChaudhuriDOP19} where 
LL formulas are attached to worlds, representing  constraints or modalities.
We shall elaborate more on HyLL in Section \ref{sec:discussion}, since the domain-aware type system in \cite{DBLP:conf/concur/CairesPPT19} is based on it. 

In this paper, the base logic is   Subexponential Linear Logic (\sell) \cite{danos93kgc,DBLP:conf/ppdp/NigamM09}. In \sell, instead of having 
a unique exponential $!$, we may have 
infinitely many \emph{colored} subexponentials of the form $\nbang{a}$. This kind of modalities have been extensively used 
to 
interpret  different concurrent languages featuring behaviors
as spatial, epistemic and temporal modalities \cite{DBLP:journals/tcs/OlartePN15}. For that, 
following the 
\emph{computation-as-proof-search}
interpretation \cite{miller87lics}, 
it has been shown that 
(concurrent/modal) programs correspond to \sell\ formulas and their execution to proof search.
  However, so far, we do not have yet a \emph{formulas-as-types} interpretation for \sell, and this is one of the main contributions of this paper.

\noindent\textbf{Session typed $\pi$-calculi.}
 Session types is a formalism to reason about  processes,   providing a typing discipline for concurrent programming based on process algebras such as the $\pi$-calculus~\cite{DBLP:conf/parle/TakeuchiHK94,DBLP:conf/esop/HondaVK98}. They have been deeply investigated  in several contexts, as, for instance: detection of subtle errors in implementations;  designing and analysis of choreographies between distributed systems as in multiparty session types;   expressiveness  power to model communication properties such as deadlock-freedom, abortable computations,  non-determinism, security, etc (see~\cite{DBLP:conf/wsfm/Dezani-Ciancaglinid09, DBLP:journals/mscs/CapecchiCD16, DBLP:conf/sfm/CoppoDPY15,DBLP:conf/forte/NantesP18}).
The work of Caires, Pfenning and Toninho~\cite{DBLP:journals/mscs/CairesPT16}  proposes a connection between a type system for the synchronous $\pi$-calculus and intuitionistic linear logic that provides a correspondence between cut elimination steps and process reductions, thus extending the Curry-Howard isomorphism to linear propositions as session types. Since then, other extensions based on variants of the $\pi$-calculus were proposed, e.g., for the asynchronous $\pi$-calculus~\cite{deyoung:LIPIcs:2012:3675} and  for the $\pi$-calculus with constructs for non-determinism and abortable behavior~\cite{DBLP:conf/esop/CairesP17}.

 \noindent\textbf{Contributions.} Inspired by a recent connection between session typed processes and  HyLL~\cite{DBLP:conf/concur/CairesPPT19}, we propose a novel  interpretation of   \sell\ via a session typed $\pi$-calculus with explicit domains. The interest of this extension is justified by the expressive power of linear logic and its two extensions \sell\ and HyLL \cite{DBLP:journals/mscs/ChaudhuriDOP19}:  whereas HyLL can be encoded in (vanilla) linear logic, the subexponentials in 
\sell\ cannot be directly encoded in LL 
in its standard presentation \cite{DBLP:conf/lpar/LellmannOP17}. 

Differently from~\cite{DBLP:conf/concur/CairesPPT19}, our processes are decorated with their current locations, and at each step, the domain/location of each end-point of the communication is explicit. 
Interaction may happen from different locations and, if needed, a process can request to migrate to a common secure location. This gives an extra  expressiveness to the language, allowing for reasoning about domains, migration and accessibility at the level of processes.
We thus propose a session type discipline for a domain-aware synchronous $\pi$-calculus with mobility and choice, that provides  yet another extension of the Curry-Howard interpretation. 
As we shall see, we interpret subexponentials as domains, which allows us to express and enforce where the processes reside and where communication 
must happen, respecting the accessibility constraints embedded in \sell's    subexponential preorder. We prove that  
cut-reductions correspond to process reductions (Th. \ref{th:cut-elim}). Moreover,  the proof of the identity expansion theorem (Th. \ref{th:id-exp}) reveals the existence of  process mediators between the different locations of the system. As a corollary of the cut elimination theorem, we prove other properties of the system as type preservation and global progress. 

\noindent\textbf{Organization.}  Section~\ref{sec:preliminar} presents the
basic notions  regarding LL and \sell\  that are necessary for understanding
this work. In Section~\ref{sec:distri-py} we introduce the language of
processes that is adopted, as well as the associated structural congruence
identities and operational semantics. Section~\ref{sec:curry-howard} presents
the proposed Curry-Howard interpretation with subexponentials as domains. We
prove relevant properties  as cut-elimination and identity expansion, and
provide   examples to illustrate our framework.  Section~\ref{sec:discussion}
concludes with a discussion about the results presented and compares our
approach with the one in \cite{DBLP:conf/concur/CairesPPT19}. 
\section{Preliminaries}\label{sec:preliminar} 

\newcommand{\varp}{\mathtt{p}}

In this section  we review some of the  basic proof theory of linear logic (see \cite{troelstra92csli} for more details) and its extension with subexponentials \cite{danos93kgc}.

Intuitionistic Linear logic (ILL) ~\cite{girard87tcs} is a resource conscious logic, in the sense that formulas are consumed when used during proofs, unless they are marked with the exponential $!$. Formulas marked with $!$ behave classically, i.e., they can be contracted (duplicated) and weakened (erased) 
 during proofs.

\begin{figure}
\centering\resizebox{.65\textwidth}{!}{$
\begin{array}{c}
 \infer[\tensor_L]{\Gamma, F \tensor H \vdash G}
{\Gamma, F, H \vdash G} 
\quad 
\infer[\tensor_R]{\Gamma_1, \Gamma_2 \vdash F \tensor H}
{\Gamma_1 \vdash F & \Gamma_2 \vdash H}
\qquad
 \infer[\with_{L_i}]{\Gamma, F_1 \with F_2 \vdash G}
{\Gamma, F_i\vdash G} 
\quad 
\infer[\with_R]{\Gamma \vdash F \with H}
{\Gamma \vdash F & \Gamma \vdash H}
\\\\
\infer[\lolli_L]{\Gamma_1, \Gamma_2, F \lolli H \vdash G}
{\Gamma_1 \vdash F & \Gamma_2, H \vdash G}
\quad 
\infer[\lolli_R]{\Gamma \vdash F \lolli H}{\Gamma, F \vdash H}
\qquad
\infer[\plus_L]{\Gamma, F \plus H \vdash G}
{\Gamma, F \vdash G & \Gamma, H \vdash G}
\quad 
\infer[\plus_{R_i}]{\Gamma \vdash F_1 \plus F_2}{\Gamma \vdash F_i}
\\\\
\infer[\one_L]{\Gamma, \one \vdash G}
{\Gamma \vdash G}
\quad 
\infer[\one_R]{ \cdot \vdash \one}{}
\quad 
\infer[\zero_L]{ \Gamma,\zero\vdash G}{}
\quad 
\infer[\top_R]{ \Gamma\vdash \top}{}
\quad
 \infer[I]{\varp \vdash \varp}{} 
\qquad 
\end{array}
$
}

\caption{Propositional fragment of iMALL (intuitionistic linear
logic without exponentials). }
\label{fig:ll}
\vspace{-2mm}
\end{figure}

Formulas in ILL  (without exponentials) are built from the grammar below: \\
$$
 F ::= \varp \mid \zero \mid \one \mid \top \mid  F_1 \tensor F_2 \mid F_1 \oplus F_2 \mid F_1 
 \lolli F_2 \mid F_1 \with F_2 
$$

\noindent where $\varp$ is an atomic proposition. 
The connectives are: 
the \emph{multiplicative} versions  of true (\one), 
conjunction $\otimes$ and implication $\limp$; and the \emph{additive} 
versions of conjunction $\with$,
disjunction $\oplus$, false $\zero$
and true $\top$. 

The proof rules are in Figure~\ref{fig:ll}.
Note that the multiplicative rules split the  context ($\otimes_R$, $\limp_L$) while the additive rules share the context in the premises ($\with_R$, $\oplus_L$). 

Contraction and weakening of formulas  are controlled by the 
connective $\bang$, called bang, whose inference rules are: 

{{$\small
\qquad\qquad \qquad\begin{array}{c}
  \infer[\bang_L]{\Gamma, \bang F \vdash G}{\Gamma, F \vdash G}
  \qquad
  \infer[\bang_R]{\bang \Gamma \vdash \bang G}{\bang \Gamma \vdash G}
  \qquad
    \infer[W]{\Gamma, \bang F \vdash G}{\Gamma \vdash G}
  \qquad 
  \infer[C]{\Gamma, \bang F \vdash G}{\Gamma, \bang F, \bang F \vdash G}
 \end{array}
$}}

Notice that one is  allowed to introduce a $\bang$ on the right ($!_R$) only if  all formulas 
in the context are  marked with a bang 
($!\Gamma = \{!F \mid F \in \Gamma\}$). 
On the other side, the left rule does not impose any restriction. 
  The rule  $\bang_R$ is commonly called
\emph{promotion}, while the rule $\bang_L$ is  called \emph{dereliction}.

\subsection{Subexponentials}

Intuitionistic linear logic  with subexponentials (\sell) shares with ILL   all connectives except the exponentials:  Instead of having a
single exponential $\bang$, \sell\ may contain as many colored/labelled bangs (e.g., $\nbang{{a}}$)
as needed.  Such labelled bangs are called \emph{subexponentials}
\cite{danos93kgc,DBLP:conf/ppdp/NigamM09}.

The proof system for intuitionistic \sell\ is parametric on  a
\emph{subexponential signature} $\Sigma = \tup{I, \preceq, U}$, where $I$ is a
set of labels, $U \subseteq I$ is a set specifying 
the \emph{unbounded}  subexponentials, i.e. those that  allow weakening and contraction (the subexponentials in  $I \setminus U$ are called \emph{linear}), and $\preceq$ is a pre-order among the elements of
$I$. 
For cut-elimination \cite{DBLP:conf/ppdp/NigamM09}, it is assumed that $U$ is upwardly closed w.r.t. $\preceq$, \ie, if $u \in U$ and $u
\preceq u'$, then $u' \in U$. 
We shall use 
$s,a,b,c\cdots$ to range over elements in $I$ and  $u,u',w,w',\cdots$ to denote unbounded subexponentials. 

The system $\sell$ is constructed by extending the rules of ILL in Fig. \ref{fig:ll} with the following ones. 
For each $a \in I$, we add the   rules
corresponding to dereliction and promotion:

\[\small
\infer[\nbang{a}_L]{\Gamma, \nbang{a} F \vdash G} {\Gamma, F \vdash
G} \qquad   \quad
\infer[\nbang{a}_R]{\nbang{a_1} F_1, \ldots \nbang{a_n} F_n
\vdash \nbang{a} G}{\nbang{a_1} F_1, \ldots \nbang{a_n} F_n \vdash  G \qquad a \preceq  {a_i}} 
\]

\noindent
Note the side condition  in  $\nbang{a}_R$:  One can only introduce a $\nbang{a}$   if
the  formulas in the left-context are all marked with indices  greater or equal than $a$. Hence,  
if $b$ is accessible from $a$ ($a \preceq b$),   we can  show that
$\nbang{b}F \vdash \nbang{a}F$. However, if $a$ and $b$ are
not related, it is not possible to prove $\nbang{a}F \equiv
\nbang{b}F$  for an
arbitrary formula $F$, where $F\equiv G$ denotes logical equivalence:  $(F \limp G) \with (G \limp F)$.  This means that the subexponentials are not canonical \cite{danos93kgc,DBLP:conf/ppdp/NigamM09}
and the subexponential signature  determines the provability relation in \sell. 
This fact has been used to interpret the subexponentials in different ways, 
ranging from (nested) spaces in distributed systems, time-units, the epistemic state of agents and preferences \cite{DBLP:journals/tcs/OlartePN15,DBLP:journals/tplp/PimentelON14,DBLP:journals/tcs/NigamOP17}.

\noindent\textbf{Dyadic System.}
  In this paper  we shall use the so-called dyadic system \cite{DBLP:journals/logcom/Andreoli92} that 
results from the incorporation of the structural rules of weakening and contraction  into the introduction rules. The key observation is that, on the left of the sequent, 
a formula $\nbang{u}{F}$ can be contracted and weakened whenever $u$ is unbounded (i.e., $u\in U$). This is reflected
  in the syntax of dyadic sequents of the form $\Gamma:\Psi \vdash G$ which have two contexts, namely, a subexponential context 
  containing formulas marked with a bang 
  ($\Gamma = \{\nbang{a_1}F_1,...,\nbang{a_n}{F_n}\}$) 
  and
  the linear context ($\Psi$). As an example, consider the  following  rules: \footnote{The reader familiar with dyadic systems for linear logic may wonder why we keep the bang in the classical context. The reason is to simplify the notation of the type system introduced in the next section.}\\
\begin{center}  
$\small
  \infer[\nbang{a}{}_L]{\dseq{\Gamma,\nbang{a}{F}}{\Psi}{G}}
  {\dseq{\Gamma}{\Psi,F}{G}}
  \qquad
  \infer[\textit{copy}]{\dseq{\Gamma,\nbang{u}{F}}{\Psi}{G}}
  {\dseq{\Gamma,\nbang{u}{F}}{\Psi,F}{G}}
  \qquad
  \infer[\textit{store}]{\dseq{\Gamma}{\Psi, \nbang{a}{F}}{G}}{\dseq{\Gamma,\nbang{a}{F}}{\Psi}{G}}
  \qquad
  \infer[!_R]{\dseq{\Gamma}{\cdot}{\nbang{a}G}}{
  \dseq{\Gamma^{\succeq a}}{\cdot}{G}
  }
 $
\end{center}
  In $\nbang{a}{}_L$, $a\not\in U$ (i.e., $a$ is a linear
  subexponential); in $\textit{copy}$, $u\in U$ (and $F$ is
  contracted);
The \textit{store} rule moves marked formulas to the subexponential context. 
The promotion rule $!_R$   has the proviso  
that there are no formulas marked with  $\nbang{c}{}$ if  $c\not\in U$ or $a\not\preceq c$.
Note that this must be the case since linear subexponentials cannot be weakened (nor contracted). 
The notation $\Gamma^{\succeq a}$ will be used to denote the following 
subset of $\Gamma$:  
  $\{\nbang{b}{F} \in \Gamma \mid a\preceq b \}$ and it is defined if
  for all linear $c$,  
  $\nbang{c}{F} \in \Gamma$ implies  $a \preceq c$. 
Therefore,   in $\Gamma^{\succeq a}$, all the unbounded subexponentials not related with $a$ are weakened, 
and the remaining subexponentials must be greater than $a$, 
thus reflecting the behavior of the promotion rule.

\subsection{Subexponential Quantifiers}
The universal ($\forallLoc$) 
and existential ($\existsLoc$) quantifiers
for subexponentials \cite{DBLP:journals/tcs/NigamOP17}  require  
a  typing information
to guarantee that the cut rule is admissible (Theorem \ref{th:cut} below). 
More precisely, 
given a subexponential signature
$\Sigma=\tup{I, \preceq, U}$, the judgment $\Typeloc{b}{a}$ is true whenever  $b \preceq a$. Hence, we obtain the following set of typed 
\emph{subexponential constants}:
$
\aprec_\Sigma=\{\Typeloc{b}{a}\,\mid\,a,b\in I, b\preceq a\}
$. 
Similar to the universal quantifier $\forall$, that introduces \emph{eigenvariables}
to the signature, the rule for  universal quantification
on the right (and the rule for existential on the left)
introduces a (fresh) \emph{subexponential variable} $\Typeloc{\alpha}{a}$, where
$a$ is a subexponential constant, \ie, $a \in I$. Thus, \sellU\ sequents have the form 
$\aprec~:~ \Gamma~:~ \Psi\vdash G$, 
where 
$
\aprec = 
\aprec_\Sigma\cup\{\Typeloc{\alpha_1}{a_1}, \ldots, \Typeloc{\alpha_n}{a_n}\}$,   
$\{\alpha_1, \ldots, \alpha_n\}$
is a  set of subexponential variables
  and $\{a_1, \ldots, a_n\} 
\subseteq I$ is a set of subexponential constants.   We shall use $I(X)$ to denote the set of constants and subexponential variables. 

Formulas in \sellU\ 
are as before with the addition of   $
   \forallLoc \Typeloc{\alpha}{s}. F$ and $\existsLoc \Typeloc{\alpha}{s}. F
$
where $\alpha: s$ is a (typed) subexponential variable,  and $s$ is 
a subexponential index (i.e., $s \in I(X)$).

The introduction rules for the subexponential quantifiers look similar to those
introducing the first-order quantifiers (see 
Fig. \ref{fig:sellU}). 
Intuitively, subexponential variables play a similar role as eigenvariables. 
The generic variable $\Typeloc{\alpha}{a}$ represents \emph{any subexponential --
constant or variable --} that is in the ideal of $a$ (i.e., in the set $\{c \mid c \preceq a\}$). Hence it can be substituted 
by any  subexponential $s$ of type $b$, with 
$b\preceq a$.

The other rules of the (dyadic) system for \sellU\  in Figure \ref{fig:sellU} are similar to those in Fig. \ref{fig:ll}. Note that the context with unbounded subexponentials ($\Upsilon$)
is duplicated while those with linear subexponentials ($\Delta$) are split in the premises of multiplicative rules. 
When the context $\mathcal{A}$ is unimportant or it can be inferred from the context, we will omit it. 


\begin{figure}
\resizebox{\textwidth}{!}
{$
\begin{array}{c}
 \infer[\tensor_L]{\dseq{\Gamma}{\Psi, F\otimes F'}{G}}{\dseq{\Gamma}{\Psi, F, F'}{G}}
\quad
\infer[\tensor_R]{
\dseq{\Upsilon,\Delta, \Delta' }{\Psi, \Psi'}{F \tensor F'}}
{
\dseq{\Upsilon,\Delta}{\Psi}{F}
&
\dseq{\Upsilon,\Delta'}{\Psi'}{F'}
}
\qquad
 \infer[\with_{L_i}]{
 \dseq{\Gamma}{\Psi, F_1\with F_2}{G}}
 {
 \dseq{\Gamma}{\Psi, F_i}{G}
 }
\quad
\infer[\with_R]{
\dseq{\Gamma}{\Psi}{F\with F'}}
{
\dseq{\Gamma}{\Psi}{F}
&
\dseq{\Gamma}{\Psi}{F'}
}
\\\\
\infer[\lolli_L]{
\dseq{\Upsilon, \Delta, \Delta'}{\Psi, \Psi', F\limp F'}{G}}
{
\dseq{\Upsilon, \Delta}{\Psi }{F}
&
\dseq{\Upsilon, \Delta'}{\Psi', F' }{G}
}
\quad
\infer[\lolli_R]{
\dseq{\Gamma}{\Psi}{F \limp F'}
}
{
\dseq{\Gamma}{\Psi,F}{F'}
}
\qquad
\infer[\plus_L]{
\dseq{\Gamma}{\Psi, F\oplus F'}{G}
}
{
\dseq{\Gamma}{\Psi, F}{G}
&
\dseq{\Gamma}{\Psi, F'}{G}
}
\quad
\infer[\plus_{R_i}]{
\dseq{\Gamma}{\Psi}{F_1 \oplus F_2}
}
{
\dseq{\Gamma}{\Psi}{F_i}
}
\\\\
\infer[\forallLoc_L]{
\dseq{\Gamma}{\Psi, \forallLoc\Typeloc{\alpha}{a}.F}{G}
}{
\dseq{\Gamma}{\Psi, F[l/\alpha]}{G}
}
\qquad
\infer[\forallLoc_R]{
\dseqA{\Gamma}{\Psi}{\forallLoc\Typeloc{\alpha}{a}.F}
}
{
\dseqAvar{(\mathcal{A}, \Typeloc{\alpha_e}{a})}{\Gamma}{\Psi}{F[\alpha_e/\alpha]}
}
\qquad
\infer[\existsLoc_L]{
\dseqA{\Gamma}{\Psi,\existsLoc\Typeloc{\alpha}{a}.F}{G}
}
{
\dseqAvar{(\mathcal{A},\Typeloc{\alpha_e}{a})}{\Gamma}{\Psi,F[\alpha_e/\alpha]}{G}
}
\qquad
\infer[\existsLoc_R]{
\dseq{\Gamma}{\Psi}{\existsLoc\Typeloc{\alpha}{a}.G}
}{
\dseq{\Gamma}{\Psi}{G[l/\alpha]}
}
\\\\
  \infer[!_L]{\dseq{\Gamma,\nbang{a}{F}}{\Psi}{G}}
  {\dseq{\Gamma}{\Psi,F}{G}}
  \qquad
  \infer[copy]{\dseq{\Gamma,\nbang{u}{F}}{\Psi}{G}}
  {\dseq{\Gamma,\nbang{u}{F}}{\Psi,F}{G}}
\qquad
  \infer[store]{\dseq{\Gamma}{\Psi, \nbang{a}{F}}{G}}{\dseq{\Gamma,\nbang{a}{F}}{\Psi}{G}}
  \qquad
  \infer[!_R]{\dseq{\Gamma}{\cdot}{\nbang{a}G}}{
  \dseq{\Gamma^{\succeq a}}{\cdot}{G}
  }
\\\\
\infer[\one_L]{
\dseq{\Gamma}{\Psi,\one}{G}
}
{
\dseq{\Gamma}{\Psi}{G}
}
\quad
\infer[\one_R]{
\dseq{\Upsilon}{\cdot}{\one}
}
{
}
\quad
\infer[\zero_L]{
\dseq{\Gamma}{\Psi,\zero}{G}
}{}
\qquad
\infer[\top_R]{ 
\dseq{\Gamma}{\Psi}{\top}}{}
\quad
 \infer[I]{
 \dseq{\Upsilon}{\varp}{\varp}
 }
 {}
\end{array}
$
}

\caption{Dyadic system for \sellU. $\Upsilon$
denotes a context with  formulas
of the form $\nbang{u}{F}$ where $u\in U$ (unbounded subexp.). $\Delta $ and $\Delta'$ contain only 
 formulas of the form 
$\nbang{a}{F}$ where $a\notin U$ (linear subexp.). $\Gamma$ denotes a context that may contain both, linear and unbounded subpexponentials.
$\Psi$ is a multiset of formulas. 
In $!_L$, $a$ is linear
and $u$ is unbounded in  $\textit{copy}$. 
In the rules for $\forallLoc$
and $\existsLoc$,   $\Typeloc{l}{b} \in \aprec$,
$b\preceq a$ 
 and $\alpha_e$ is fresh. Since the context $\mathcal{A}$ is only modified in rules $\forallLoc_R$ and $\existsLoc_L$, it is omitted in the other rules. 
}
\label{fig:sellU}
\end{figure}

\begin{theorem}[Cut-elimination \cite{DBLP:journals/tcs/NigamOP17}]\label{th:cut}
The  cut-rule below is admissible in \sellU for any (upwardly closed w.r.t. $U$) signature.
$$
\infer[Cut]{\dseqA{\Upsilon,\Delta,\Delta'}{\Psi,\Psi'}{G}}{
\dseqA{\Upsilon,\Delta}{\Psi}{F}
&
\dseqA{\Upsilon,\Delta'}{\Psi',F}{G}
}
$$
\end{theorem}
\section{The Language of Processes}\label{sec:distri-py}

\newcommand{\piopenout}[4]{{#1}_{#2}\langle{#3}\rangle.{#4}}
\newcommand{\piopenmove}[4]{{#1}_{#2}.{\tt move}\langle{#3}\rangle.{#4}}
\newcommand{\redrule}[1]{[{\tt #1}]}

We introduce the syntax and operational semantics of the synchronous
$\pi$-calculus extended with binary choice \cite{DBLP:journals/mscs/CairesPT16} and explicit domain location,  migration and communication, similarly to  \cite{DBLP:conf/concur/CairesPPT19}. Concrete domains/locations range over $a,b,c$,  domain/location variables over $\alpha,\alpha',\beta$ and channel names  over $x,y,z$.   Domain identifiers are attached to processes $(P,Q,\ldots)$, to explicitly identify their current location.
Processes are defined in terms of pre-processes as below.

\begin{definition}[Processes]\label{def:processes} Given countably infinite
 (disjoint) sets, $\mN$  of names $(x,y,z,u,v)$,  $I$  of concrete
 domains ($a,b,c,d$)  and  $\mV$ of domain variables
 ($\alpha,\alpha',\beta$),
{\em processes} are built from: \\

$
\begin{array}{lllll}
\mbox{Pre-process} & \quad & p,q & ::= &  \skipp \mid (p \parallel q)  \mid  \pioutput{x}{a}{y}{p} \mid \piinput{x}{a}{y}{p} \mid !\piinput{x}{a}{y}{p} \mid \pinew{x}{p}\\
&&&& \picase{x}{a}{p}{q} \mid \piinl{x}{a}{p}{} \mid \piinr{x}{a}{p}{} \\
&&&& \pimoveout{x}{a}{y}{b}{p} \mid \pimovein{x}{a}{y}{b}{p} \\&&&& 
\piwout{x}{a}{b}{p} \mid \piwin{x}{a}{\alpha}{p}\mid \pimed{x}{a}{b}{y}
\\
\mbox{Processes} & \quad & P,Q & ::= &  \boxp{a}{p} \mid \zero \mid ( P \parallel Q) \mid \pinew{x}{P}
\end{array}
$
\end{definition}

Process $P=\boxp{a}{p}$
 denotes  a  pre-process $p$  that  is currently in location $a$. Moreover, 

\begin{itemize}
\item   Process $\boxp{b}{\pioutput{x}{a}{y}{p}}$ denotes a process in location $b$ that outputs  $y$ 
 through channel $x$ and continues as $p$. The dual process $\boxp{a}{\piinput{x}{b}{y}{p}}$ denotes a process in location $a$ that expects to receive $y$ 
 via $x$ and proceeds as $p$. 
\item Process $\boxp{b}{ \picase{x}{a}{p}{q}}$ in location $b$ offers the choice of continuing as  $p$ or $q$ 
to another process
via session $x$. The dual process $\boxp{a}{\piinl{x}{b}{p}{}}$ (resp. $\boxp{a}{\piinr{x}{b}{p}{}}$) in location $a$ chooses the process on the left (resp. right)  offered 
on session $x$
 and continues as $p$.
\item Process $\boxp{b}{\pimoveout{x}{a}{y}{c}{p} }$ in location $b$ denotes the pre-process that signals to another process
 the location  $c$ to which the pre-process $p$ should be migrated to. After the migration, the communication will continue on the fresh  session  channel $y$. The dual process $\boxp{a}{ \pimovein{x}{c}{y}{c}{p}}$ in location $a$ is prepared to migrate the communication actions in $p$  to session $y$ on a location $c$.

\item Processes $\boxp{b}{\piwout{x}{a}{b}{p}}$ and $\boxp{a}{\piwin{x}{b}{\alpha'}{p}}$  are dual and denote output and input, respectively, of  locations via  session $x$. 
\item Pre-process  $\pimed{x}{a}{b}{y}$ is a copycat
or mediator (see \S \ref{th:id-exp} below) that 
identifies/connects the channel $x$ on location $b$ with the channel $y$ on location $a$. 

\item Pre-process $!\piinput{x}{a}{y}{p}$ denotes
a replicated/persistent  input, useful for defining services that can be used as many times as needed. 
\item  As usual, $\skipp$,  $p\parallel q$   
and $\pinew{x}{P}$ 
 denote, resp.,  inaction, parallel composition and  name restriction.  
\end{itemize}
As in~\cite{DBLP:journals/tcs/Sangiorgi96a, DBLP:journals/tcs/Boreale98}, 
the objects in outputs are fresh names. Hence, in
$\pioutput{x}{a}{y}{p}$  and  $\pimoveout{x}{a}{y}{b}{p}$, the output
$y$ is a fresh name. Moreover, input names are bound in the
  body of a pre-process, i.e. $y$ in $p$ is bound in all $\piinput{x}{a}{y}{p}$,
  $!\piinput{x}{a}{y}{p}$, $\pimovein{x}{a}{y}{b}{p}$ (and $\alpha$ in
  $p$ in $\piwin{x}{a}{\alpha}{p}$).  Finally, name $x$ in process $P$
is bound in $\pinew{x}{P}$.

\begin{example}[Locations]\label{ex:st.values} The following process illustrates  the parallel composition of processes in different locations, one in $sv$ (server location) and the other in $c$ (client location) that aim to communicate with each other via  a restricted channel $x$:
$P =  \pinew{x}{(
\boxp{sv}{\piinput{x}{c}{y}{\pioutput{x}{c}{z}{\piinput{y}{c}{k}{\zero}}}}
\parallel
\boxp{c}{\pioutput{x}{sv}{y}{( \pioutput{y}{}{k}{\zero} \parallel  \piinput{x}{sv}{z}{\zero})}}
)}$
Notice that the process on the left aims to use  the name $y$ received first in a later communication, after communicating $z$ on channel $x$.
\end{example}

\begin{figure}
\hrule 
\smallskip
Structural Congruence for Pre-Processes
\smallskip
\hrule 
\smallskip
\noindent\resizebox{.95\textwidth}{!}{
$
\begin{array}{c@{\hspace{.5cm}}c@{\hspace{.5cm}}c}
p\parallel \skipp  \equiv {p} & {p\parallel q}\equiv {q\parallel p} & {p\parallel (q \parallel r)}\equiv {(p\parallel q)\parallel r)}\\
{\pinew{x} \zero}\equiv
 {\zero}&{\pinew{x}\pinew{y}p}\equiv{
                   \pinew{y}\pinew{x} p} &   p\equiv_\alpha q \Rightarrow {p}\equiv {q}\\
{p\parallel \pinew{y}{q}}\equiv {\pinew{y}{(p\parallel q)}}, \text{ if } y\notin fn(p)&\\
\end{array}
$
}
\smallskip
\hrule 
\smallskip
Structural Congruence for  Processes
\hrule 
\smallskip
\noindent\resizebox{1\textwidth}{!}{
$
\begin{array}{c@{\hspace{.8cm}}c@{\hspace{.8cm}}c}
\boxp{a}{\skipp } \equiv \zero & \boxp{a}{p\parallel q}\equiv \boxp{a}{p} \parallel \boxp{a}{q}& \boxp{a}{\pinew{y}p}\equiv \pinew{y}\boxp{a}{p}\\
 P\parallel \zero\equiv P  &\pinew{x} \zero\equiv \zero&   \boxp{a}{p}\parallel \boxp{a}{\pinew{y}q}\equiv \pinew{y}\boxp{a}{p\parallel q}, \text{ if } x\notin fn(p) \\
 P\parallel Q\equiv Q\parallel P    & \pinew{x}\pinew{y}P\equiv \pinew{y}\pinew{x} P &  P\equiv_\alpha Q \Rightarrow P\equiv Q  \\ 
P\parallel (Q\parallel R)\equiv (P\parallel Q) \parallel R &  P \parallel \pinew{x}Q \equiv \pinew{x}(P\parallel Q), \text{ if } x\notin fn(P) &\boxp{a}{\pimed{x}{a}{b}{y}}\equiv \boxp{b}{\pimed{y}{b}{a}{x}}
\end{array}
$
}

\hrule 
\smallskip
\resizebox{0.95\textwidth}{!}{
$
\begin{array}{l@{\hspace{2cm}}rcl}
\redrule{Comm}&\boxp{a}{\pioutput{x}{b}{y}{p}}\parallel \boxp{b}{\piinput{x}{a}{z}{q}}&\rightarrow&\pinew{y}{( \boxp{a}{p}\parallel \boxp{b}{q\{y/z\}})}\\
\redrule{!Comm}&\boxp{a}{\pioutput{x}{b}{y}{p}}\parallel \boxp{b}{!\piinput{x}{a}{z}{q}}&\rightarrow&\pinew{y}{( \boxp{a}{p}\parallel \boxp{b}{q\{y/z\}}) \parallel \boxp{b}{!\piinput{x}{a}{z}{q}}}\\
\redrule{LocOut}&\boxp{a}{\piwout{x}{b}{c}{p}}\parallel \boxp{b}{\piwin{x}{a}{\alpha}{q}}&\rightarrow&\boxp{a}{p}\parallel\boxp{b}{ q\{c/\alpha\}}\\
\redrule{Move}&\boxp{a}{\pimoveout{x}{b}{y}{c}{p}}\parallel \boxp{b}{\pimovein{x}{a}{y'}{c}{q}}&\rightarrow&\pinew{y}{\boxp{c}{p\parallel q\{y/y'\}}}\\
\redrule{LCase}&\boxp{a}{\piinl{x}{b}{p}{}}\parallel \boxp{b}{\picase{x}{a}{q_1}{q_2}} &\rightarrow&\boxp{a}{p}\parallel \boxp{b}{q_1}\\
\redrule{RCase}&\boxp{a}{\piinr{x}{b}{p}{}}\parallel \boxp{b}{\picase{x}{a}{q_1}{q_2}} &\rightarrow&\boxp{a}{p}\parallel \boxp{b}{q_2}\\
\redrule{Id } & \pinew{x}{ (\boxp{a}{\pimed{x}{a}{b}{y}}\parallel P)} &\rightarrow & P\{y/x\}\\
\redrule{Res}&\mbox{ if } P\rightarrow  Q \mbox{ then\quad } (\nu y)P &\rightarrow&  (\nu y)Q\\
\redrule{Par}&\mbox{ if } Q\rightarrow Q'  \mbox{ then\quad} P\parallel Q &\rightarrow& P \parallel Q'
\end{array}
$
}
\hrule
\caption{Structural congruence and operational semantics\label{fig:sos}}
\end{figure}

Communication will rely on a  location sensitive  extension of the structural congruence and the operational semantics.
The structural congruence identities are defined at the level of processes and pre-processes and 
the operational semantics is defined by the reduction $\rightarrow$
(closed under $\equiv$) in  Fig. \ref{fig:sos}, where
  $q\{y/x\}$ denotes the implicit substitution of $y$ for $x$ in $q$
  (also for location substitution in rule $\redrule{LocOut}$).

Our semantics is similar  to the standard  operational semantics of
the $\pi$-calculus. The difference here is that the location of each
process is explicit and one process in one location can communicate
with a process in another location. For instance, rule
$\redrule{Comm}$ describes the communication of a process in location
$a$ that sends a fresh name $y$ to a process in location $b$ through
channel $x$. After the communication, the two processes share the
restricted  channel $y$. Note that  both
processes continue  in their original location after the interaction. The rule $\redrule{Move}$ allows processes 
to migrate to a \emph{known a-priori} location $c$
and the communication continues through the fresh channel $y$. Using the rule  $\redrule{LocOut}$,
it is possible to communicate a location before migrating to it. 
The other rules are standard.

\begin{example}\label{ex:st.values-red}
Consider the process $P$ in 
Example~\ref{ex:st.values}. Note that 
$P \rightarrow^*  \pinew{x}\zero \equiv \zero$, since: 

$
\small
\begin{array}{rl}
\boxp{sv}{\piinput{x}{c}{y}{\pioutput{x}{c}{z}{\piinput{y}{c}{k}{\zero}}}}
\parallel
\boxp{c}{\pioutput{x}{sv}{y}{(\pioutput{y}{sv}{k}{\zero} \parallel \piinput{x}{sv}{z}{\zero})}}
& \rightarrow \pinew{y}{(
\boxp{sv}{\pioutput{x}{c}{z}{\piinput{y}{c}{k}{\zero}}}
\parallel
\boxp{c}{(\pioutput{y}{sv}{k}{\zero} \parallel \piinput{x}{sv}{z}{\zero})}
)}
\\
& \rightarrow \pinew{y,z}{(
\boxp{sv}{\piinput{y}{c}{k}{\zero}}
\parallel
\boxp{c}{\pioutput{y}{sv}{k}{\zero}}
)}
  \rightarrow \pinew{y,z,k}{(
\boxp{sv}{\zero}
\parallel
\boxp{c}{\zero}
)}\equiv \zero
\end{array}
$
{The intended communications are performed: first the pre-process on the left receives name $y$ through $x$, then it communicates $z$ through the same channel, and finally it uses the name $y$ received initially. Observe that the left pre-process has ``stored'' the name $y$ for future use.}
\end{example}

\section{Subexponentials as Domains: Curry-Howard Interpretation}\label{sec:curry-howard}
In this section we introduce a type system for 
the domain-aware process calculus presented in the previous section. We show that operational reductions correspond to proof reductions in the process of cut-elimination. Also, the proof of the identity expansion theorem (\S \ref{th:id-exp}) reveals the existence of process mediators among the different locations of the system. 
We start with the definition of types.
\begin{definition}[Types]\label{def:processes-types}
Types are  built from formulas in \sellU as follows: 
\[
\begin{array}{lllll}
\mbox{Pre-types} & \qquad &  f,g& ::= & \one  \mid \varp \mid f\otimes g\mid f \lolli g
 \mid f \oplus g \mid f \with g \mid 
  \nbang{s} f \mid  \forallLoc \Typeloc{l_x}{a}. f \mid \existsLoc \Typeloc{l_x}{a}. f\\
\mbox{Location-sequences} & &  w,w' & ::= &  s \mid (w,w') \mid w . w'  \mid \lambda x. w
\\
\mbox{Types} & &  A,B & ::= & \BOXF{w}{f}
\end{array}
\]
where $s$ is a subexponential (constant or variable) and $\BOXF{w}{f}$ is defined as: \\

{
\noindent\resizebox{\textwidth}{!}{
$
\begin{array}{lllllll}
\BOXF{s}{\one} &=& \nbang{s}\one 
&\quad & 
\BOXF{s}{\varp} &=& \nbang{s}\varp\\
\BOXF{a.(w,w')}{f {\star} g}&=& \nbang{a}(\BOXF{w}{f} \star \BOXF{w'}{g}) 
&\quad & \star \in \{\otimes, \limp, \oplus, \with\}
\\
\BOXF{a.w}{\nbang{c}f} & = & \nbang{a}\BOXF{c.w}{f}
\\
\BOXF{a.(\lambda y.  w)}{\forallLoc{\Typeloc{l_x}{b}}. f} &=&
\nbang{a}\forallLoc{\Typeloc{l_x}{b}}.\BOXF{w \{ l_x/y\}}{f}
&\quad & 
\BOXF{a.(\lambda y. w)}{\existsLoc {\Typeloc{l_x}{b}}.f} & = &
\nbang{a}{\existsLoc{\Typeloc{l_x}{b}}.\BOXF{w\{ l_x/y\}}f}
\end{array}
$}
}

\end{definition}

Note that in valid types,   all the sub-formulas (including the
formula itself)  are   preceded with  a bang. For instance,  the
(valid) types 
$
A = \nbang{a}{}(\nbang{b}{\one} \otimes \nbang{c}{\one})$ and $B = 
\nbang{a}{}(\nbang{b}{( \nbang{c}{\one}\otimes \nbang{d}{\one}) }~\limp~\nbang{e}{\one})
$ can be written, respectively, as  
$
A = \BOXF{a.(b,c)}{\one\otimes \one}
$ and  $B = \BOXF{a.(b.(c,d),e)}{ (\one\otimes \one) \limp \one}$. Such a definition keeps track of
all the spaces (subexponentials) needed to adequately locate where a
service is provided/offered and when processes need to migrate to
continue the communication.

The case for quantifiers requires to bind variables occurring in $w$ with the binding variable in  $\existsLoc$ or $\forallLoc$. We write $\lambda y. w$ to bind occurrences of a variable $y$ in $w$ and $w\{ l_x/y\}$  is the result of the application of $\lambda y. w$ to $l_x$.  For instance,  $A=\nbang{a}(\forallLoc{l_x:b}. \nbang{l_x}(\nbang{s}(1 ) \tensor \nbang{s}(1)) $ is written as $A=\BOXF{a.\lambda y. (y.(s, s)))}{\forallLoc{l_x:b}(1\tensor 1)$}.

We shall write 
 $\boxf{a}{f}$ to represent the type 
 $\BOXF{a.w}{f}$ where $w$ can be inferred from the context (or from a typing derivation). 
 As we already know,  all types are marked with a bang and hence,
 with  $\boxf{a}{f}$, we signal the outermost bang and leave implicit
 the inner bangs, simplifying the notation. 
 
Following  \cite{DBLP:journals/mscs/CairesPT16},   the multiplicative connectives    $\tensor,\limp$  correspond to input/output and the additive ones   $\with,\oplus$ correspond to selection/choice. 
Here, from  \sellU, 
 $\boxf{a}{}$ (a formula preceded with $\nbang{a}$) represents a service 
offered from a site/location $a$. 
Moreover, bangs will allow us to migrate services  to other locations. 
 The quantifiers $\forallLoc,\existsLoc$ 
 will be used for the communication of  locations. 

Since all  types (and subformulas)  are marked with a subexponential, the linear context will be always empty. Hence, 
different from the system in Fig. \ref{fig:sellU}, 
 type judgments will take the form  
$
\deduce{\seqtypeA{\mathcal{A}}{\Delta}{P}{T}}{}
$  where: $\Delta$ is a set of \emph{typed channels} of the form $\TYPELOC{x}{b}{a}{A}$ meaning ``the service $A$ is available 
from location $a$ on channel $x$
and it can be used/consumed from location $b$
''; $P$ is a process;  and $T$ is a typed channel. We read $P::\TYPELOC{x}{b}{a}{A}$  as
``$P$ implements the service $A$ on channel $x$ at location $a$ to be
used  by a processes located at $b$''. Processes are
  considered modulo structural congruence hence, as in
  \cite{DBLP:conf/concur/CairesPPT19}, typability is closed
  under $\equiv$ by definition.
Channel names appear at most once in the sequent (either on the right or on the left). 
As we shall see, a process  in a location $a$ can
only \emph{implement} services on that location but it can consume/use
services from any other location
accessible from $a$.  
Similarly as we did in Fig. \ref{fig:sellU}, we shall omit $\mathcal{A}$ when it is unimportant or it  is clear from the context. 

\noindent\textbf{Subexponential signature.}
We shall gradually introduce the system rules, with
  its intended interpretation.
  The complete set of rules is in
 Fig. \ref{fig-type-sys}, App. \ref{app:typesystem}. 
    We assume 
a set of linear subexponentials $I_l$
 with a preorder $\preceq_l$. 
Those indices represent the different locations of the system.
We build the subexponential signature 
$\Sigma = (I,\preceq,U)$ where: 
$
\begin{array}{llllllll}
I &=& I_l \cup  \{ \unb{a} \mid a \in I_l\} \cup \{a \sqcap b \mid a,b \in I_l\}
&\mbox{and}&
 U &=& \{ \unb{a} \mid a \in I_l\}
\end{array}
$ 

The preorder $\preceq$ satisfies:
for linear subexponentials, 
$a \preceq b$ iff $a \preceq_l b$. Moreover,   $ a \preceq \unb{a}$;
$\unb{a} \preceq \unb{b}$ iff $a \preceq b$;
$a\sqcap b = b \sqcap a$;  $a\sqcap b  \preceq a$;  and  $a = a \sqcap a$.

Intuitively, we start with a set of (linear) locations. Given one of such location $a$, the unbounded version of it (where replicated services 
in that location reside)
 is denoted as $\unb{a}$ and $a\preceq \unb{a}$.  For each two linear locations $a,b$, the linear subexponential $a\sqcap b$ denotes a location that is below 
  $a$ and $b$. 

In the following, we shall use  $\Upsilon$
to denote a context 
containing only channel types with unbounded subexponentials, i.e., 
$\TYPELOC{x}{b}{\unb{a}}{A}$
for some (linear) $a$. We shall use $\Delta$, $\Delta'$ to denote a context 
with only linear subexponentials and 
$\Gamma$ to denote a context with linear and unbounded subexponentials. 

\noindent\textbf{Unit.}
  The rules for $\one$ are the following: 

$$
\small
\qquad
\infer[\one_L]{\seqtype{\Gamma, \TYPELOC{x}{b}{a}{\one}}{P}{T}}{\seqtype{\Gamma}{P}{T}}
\quad
\infer[\one_R]{\seqtype{\Upsilon}{\boxp{b}{\zero}}{\TYPELOC{x}{a}{b}{\one}}}{}
$$

Note that in rule $\one_R$, all the formulas in $\Upsilon$  can be weakened. 
In both cases, $\boxf{a}{\one}$
is intended to be $\BOXF{a}{\one} = \nbang{a}\one$ (since the bang is already marking an atomic formula). Intuitively, the formula  $\nbang{b}{\one}$ is the type of an ending process ($\zero$) located at $b$. \\

\noindent\textbf{Multiplicative Rules.}
The rules for $\otimes$ are: 
\[
\small
\begin{array}{c}
\infer[\tensor_L]{\seqtype{\Gamma, \TYPELOC{x}{b}{a}{A \otimes B}}{\boxp{b}{\piinput{x}{a}{y}{p}}}{T}}{
\seqtype{\Gamma, {\TYPELOC{x}{b}{a}{B}}, {\TYPELOC{y}{b}{a}{A}} }{\boxp{b}{p}}{T}
}
\quad
\infer[\tensor_R]{\seqtype{\Upsilon, \Delta,\Delta'}{\boxp{b}{\pioutput{x}{a}{y}{(p \parallel q)}}}{\TYPELOC{x}{a}{b}{A\otimes B}}}
{
 \deduce{\seqtype{\Upsilon^{\succeq b},\Delta^{\succeq b}}{\boxp{b}{p}}{{\TYPELOC{y}{a}{b}{A}}}}{}
 &
  \deduce{\seqtype{\Upsilon^{\succeq b}, \Delta'^{\succeq b}}{\boxp{b}{q}}{{\TYPELOC{x}{a}{b}{B}}}}{}
}
\end{array}
\]

In the rule $\otimes_L$, the process 
$\piinput{x}{a}{y}{p}$ is waiting for an input of type $A$   from the location $a$. 
Hence, assuming that such a service $A$ is available at the fresh channel $y$, the remaining process $p$ should 
be able to continue by consuming the 
service $B$ at $x$ (and $A$ at $y$). 
In the rule $\otimes_R$, 
the output process  located at $b$ implements the  service $A\otimes B$ if  $p$ is able to type the output $y$ with the appropriate type $A$ and the continuation $q$ is able to implement the service $B$ {at} $x$. Note that no service located in $c$ s.t. $b\not\preceq c$
can be stored in the context (otherwise, 
$\Delta^{\succeq b}$ is not defined\footnote{In fact, if $\Delta$ contains only linear subexponentials and  
$\Delta^{\succeq b}$  is defined, then 
$\Delta=\Delta^{\succeq b}$ (this fact is important for the proof of Th. \ref{th:cut-elim}). We have decided to use $\Delta^{\succeq b}$ instead of $\Delta$ in order to make explicit the restriction imposed by the promotion rule.}). 
This means that a process cannot offer a service if it has not still processed all   services requested from other (non-related) locations
 (see Ex. \ref{ex:storing} below). This is an invariant and
 a design principle in our system:  a pre-process $p$ located at $b$ can only type channels on the same location (right rules). For that, the context cannot have any unrelated resource (due to the promotion rule). On the other hand, 
 left rules do not have any restriction
on the context (similar to dereliction in linear logic). However, note that 
a process located at $b$ is only able to consume a service of the form $\TYPELOC{x}{b}{a}{A}$ (offered in $a$).

We can see the rules of the type system as \emph{macro rules} built on top of rules in Fig. \ref{fig:sellU} (when  the term $P$ and the labels for channels are erased): 

\begin{center}
\hspace{2.5cm}\resizebox{.8\textwidth}{!}{$
\infer[!_L]{\dseq{\Gamma, \BOXF{a.(a. w, a. w')}{A\otimes B}}{\cdot}{G}}{
\infer[\otimes_L]{\dseq{\Gamma}{
\BOXF{a. w}{A}\otimes \BOXF{a. w'}{B}}{G}}{
 \infer=[store]{\dseq{\Gamma}{\boxf{a}{A}, \boxf{a}{B}}{G}}{
 \dseq{\Gamma,\boxf{a}{A}, \boxf{a}{B}}{\cdot}{G}
 }
}
}
\qquad
\infer[!_R]{\dseq{\Upsilon,\Delta,\Delta'}{\cdot}{\BOXF{b. (b. w,b. w')}{A\otimes B}}}{
 \infer[\otimes_R]{\dseq{\Upsilon^{\succeq b},\Delta^{\succeq b},\Delta'^{\succeq b}}{\cdot}{{\BOXF{b. w}{A}\otimes \BOXF{b. w'}{B}}}}{
 \deduce{\dseq{\Upsilon^{\succeq b},\Delta^{\succeq b}}{\cdot}{\boxf{b}{A}}}{}
 &
  \deduce{\dseq{\Upsilon^{\succeq b},\Delta'^{\succeq b}}{\cdot}{\boxf{b}{B}}}{}
 }
}
$
}
\end{center}

The rules for linear implication are  dual to those for  multiplicative conjunction. This means that an output process consumes types of the form $A\limp B$
while input processes offer services of the same shape: 
{\small
\[
\begin{array}{c}
\infer[\limp_L]{\seqtype{\Upsilon,\Delta,\Delta',\TYPELOC{x}{b}{a}{A\limp B}}{\boxp{b}{\pioutput{x}{a}{y}{(p \parallel q)}}}{T}}{
 \deduce{\seqtype{\Upsilon,\Delta}{\boxp{b}{p}}{{\TYPELOC{y}{a}{b}{A}}}}{}
 &
  \deduce{\seqtype{\Upsilon,\Delta'{,\TYPELOC{x}{b}{a}{B}}}{\boxp{b}{q}}{{T}}}{}
}
 \qquad
  \infer[\limp_R]{\seqtype{\Gamma}{\boxp{b}{\piinput{x}{a}{y}{p}}}{\TYPELOC{x}{a}{b}{A\limp B}}}{
  \seqtype{\Gamma^{\succeq b},{\TYPELOC{y}{b}{a}{A}}}{\boxp{b}{p}}{{\TYPELOC{x}{a}{b}{B}}}
  }
\end{array}
\]}

From the rule   $\limp_L$, it is clear that $\boxf{a}{(A \limp B)}$
is a formula of the shape
$\BOXF{a.(b. w, a. w')}{A \limp 
B}$. In  the rule  $\limp_R$, 
$\boxf{b}{(A \limp B)}$
is intended to mean a formula of the shape
$\BOXF{b.(a. w, b. w')}{A \limp 
B}$.
The correspondence with the  rules in Fig. \ref{fig:sellU} can be obtained similarly as we did for the rules of $\tensor$. 
It is worth noticing the annotations on channels. For instance,
in $\limp_R$, 
 the input located at $b$  receives on 
$y$ some information coming from  $a$. 
That justifies the annotation $y^b_a$. 
Dually, in $\limp_L$, the  output $y$ must be typed by the process $p$ at $b$, to be used by some process at location $a$ (annotation $y^a_b$).

\paragraph{Additive rules.}
Similar to the first systems for 
session types based on linear logic
\cite{DBLP:journals/mscs/CairesPT16}, the additive rules give meaning
to processes that offer and select different execution alternatives: 
\[
\begin{array}{c}
\infer[\oplus_L]{\seqtype{\Gamma, \TYPELOC{x}{b}{a}{A \oplus B}}{\boxp{b}{\picase{x}{a}{p}{q}}}{T}}{
\seqtype{\Gamma, {\TYPELOC{x}{b}{a}{A}} }{\boxp{b}{p}}{T}
&
\seqtype{\Gamma, {\TYPELOC{x}{b}{a}{B}} }{\boxp{b}{q}}{T}
}
\\\\
\infer[\oplus_{R1}]{\seqtype{\Gamma}{\boxp{b}{\piinl{x}{a}{p}}}{\TYPELOC{x}{a}{b}{A\oplus B}}}
{
 \deduce{\seqtype{\Gamma^{\succeq b}}{\boxp{b}{p}}{{\TYPELOC{x}{a}{b}{A}}}}{}
}
\qquad
\infer[\oplus_{R2}]{\seqtype{\Gamma}{\boxp{b}{\piinr{x}{a}{p}}}{\TYPELOC{x}{a}{b}{A\oplus B}}}
{
 \deduce{\seqtype{\Gamma^{\succeq b}}{\boxp{b}{p}}{{\TYPELOC{x}{a}{b}{B}}}}{}
}
\\\\
\infer[\with_{L1}]{\seqtype{\Gamma,\TYPELOC{x}{b}{a}{A\with B}}{\boxp{b}{\piinl{x}{a}{p}}}{T}}{
\seqtype{\Gamma,\TYPELOC{x}{b}{a}{A}}{\boxp{b}{p}}{T}
}
\qquad
\infer[\with_{L2}]{\seqtype{\Gamma,\TYPELOC{x}{b}{a}{A\with B}}{\boxp{b}{\piinr{x}{a}{p}}}{T}}{
\seqtype{\Gamma,\TYPELOC{x}{b}{a}{B}}{\boxp{b}{p}}{T}
}
\\\\
\infer[\with_R]{\seqtype{\Gamma}{\boxp{b}{\picase{x}{a}{p}{q}}}{\TYPELOC{x}{a}{b}{A \with B}}}{
\seqtype{\Gamma^{\succeq b}}{\boxp{b}{p}}{\TYPELOC{x}{a}{b}{A}}
&
\seqtype{\Gamma^{\succeq b}}{\boxp{b}{q}}{\TYPELOC{x}{a}{b}{B}}
}
\end{array}
\]
The process $\boxp{b}{\picase{x}{a}{p}{q}}$
can make use of a behavior  $A\oplus B$ 
offered at $x$ 
on location $a$, if $p$ is able to consume $A$ 
and $q$ is able to consume $B$. On the other side, 
the process $\boxp{b}{\piinl{x}{a}{p}}$
implements the behavior $A \oplus B$
by selecting $A$ (which is implemented by $p$ at $b$). 
The reasoning behind $\with$ is dual. 

\paragraph{The subexponential rules.}\label{sec:exp-rules}
The following rules govern the migration of processes to other locations.

\[
\begin{array}{c}
\infer[\bang_L]{\seqtype{\Gamma,\TYPELOC{x}{b}{a}{\nbang{c}{A}}}{\boxp{b}{\pimovein{x}{a}{y}{c}{p}}}{T}}{
\seqtype{\Gamma,\TYPELOC{y}{c}{c}{A}}{\boxp{c}{p}}{T}
}
\qquad
\infer[\bang_R]{\seqtype{\Gamma}{\boxp{b}{\pimoveout{x}{a}{y}{c}{p}}}{\TYPELOC{x}{a}{b}{\nbang{c}A}}}{
\seqtype{\Gamma^{\succeq b}}{\boxp{c}{p}}{\TYPELOC{y}{c}{c}{A}}
}
\end{array}
\]

In both rules, the processes migrate to the location $c$ and the communication continues on the fresh channel $y$. 
Note that $y$ is ``local'' to  location $c$: it is implemented and consumed inside the same location (annotation $y^c_c$). 
The corresponding derivations with the system in Fig. \ref{fig:sellU} are:

\hspace{4cm}
$
\infer[\it{def}]{\dseq{\Gamma, \BOXF{a. w}{\nbang{c}{A}}}{\cdot}{G}}{
\infer[!_L]{\dseq{\Gamma, \nbang{a}\BOXF{c. w}{A}}{\cdot}{G}}{
\infer[store]{
\dseq{\Gamma }{\BOXF{c. w}{{A}}}{G}}{
\dseq{\Gamma,\boxf{c}{A}}{\cdot}{G}
}
}
}
\qquad
\infer[\it{def}]{\dseq{\Gamma}{\cdot}{\BOXF{b. w}{\nbang{c}{A}}}}{
\infer[!_R]{\dseq{\Gamma}{\cdot}{\nbang{b}\BOXF{c\cdot w}{{A}}}}{
\dseq{\Gamma^{\succeq b}}{\cdot}{\boxf{c}{A}}
}
}
$
\\ 
\noindent\textbf{The subexponential quantifiers.}
The quantifiers will be used to   communicate locations/domains. 
{
\[
\begin{array}{c}
\infer[\existsLoc_L]{\seqtypeA{\mathcal{A}}{\Gamma,\TYPELOC{x}{b}{a}{\existsLoc_{\Typeloc{\alpha}{a\sqcap b}}.A}}{\boxp{b}{\piwin{x}{a}{\alpha}{p}}}{T}}{
\seqtypeA{\mathcal{A},\alpha\preceq a\sqcap b}{\Gamma,{\TYPELOC{x}{b}{a}{A}}}{\boxp{b}{p}}{T}
}
\qquad
\infer[\existsLoc_R]{\seqtype{\Gamma}{\boxp{b}{\piwout{x}{a}{{l}}{p}}}{\TYPELOC{x}{a}{b}{\existsLoc _{\Typeloc{\alpha}{a\sqcap b}}.A}}}{
\seqtype{\Gamma^{\succeq b}}{\boxp{b}{p}}{{\TYPELOC{x}{a}{b}{A\{l/\alpha\}}}}
}
\end{array}
\]}

 In $\existsLoc_L$, $\alpha$  is a fresh subexponential of type $a\sqcap b$. Note that this rule extends the signature with the typing information  $\Typeloc{\alpha}{a\sqcap b}$. 
In $\existsLoc_R$, $l$ must be of type $a\sqcap b$ (i.e., $l \preceq a\sqcap b$). Hence, the 
process $\boxp{b}{\piwout{x}{a}{{l}}{p}}$
implements the behavior of sending a specific location $l$, which is  reachable from both $a$ and $b$.

The rules for the universal quantifier are dual and omitted (see
Fig. \ref{fig-type-sys}, App. \ref{app:typesystem}).

\noindent\textbf{The unbounded subexponential rules.} 
In  \cite{DBLP:journals/mscs/CairesPT16}, the LL exponential $!$ is used to specify replicated servers. In our setting, 
the type for a replicated service at location $a$ will be marked with the subexponential   $\nbang{\unb{a}}$
(remember that $\unb{a}\in U$ is unbounded). A replicated service allows for using as many copies of the service as needed. This is reflected by the rules below: 
\[
\infer[copy]{\seqtype{\Gamma,\TYPELOC{x}{b}{\unb{a}}{\nbang{a}A}}{\boxp{b}{\pioutput{x}{a}{y}{p}}}{T}}{
\seqtype{\Gamma,\TYPELOC{x}{b}{\unb{a}}{\nbang{a}A}, \TYPELOC{y}{b}{a}{A}}{\boxp{b}{{p}}}{T}
}
\qquad
\infer[!^+_R]{\seqtype{\Upsilon}{\boxp{b}{!\piinput{x}{a}{y}{p}}}{\TYPELOC{x}{a}{\unb{b}}{\nbang{b}A}}}{
\seqtype{\Upsilon^{\succeq \unb{b}}}{\boxp{b}{{p}}}{\TYPELOC{y}{a}{b}{A}}
}
\]

In \emph{copy}, an instance  of $A$
is replicated and offered on the fresh channel $y$ (located at $a$). In the right rule, only unbounded resources ($\Upsilon^{\succeq \unb{b}}$) can be used to show that $p$ implements the behavior $A$ at $b$. 

As shown for $\otimes$ and $!$, the rules of the type system can be seen as macro rules of the system in Fig.~\ref{fig:sellU}: 

\begin{proposition}\label{prop:admissible}
The typing rules above are admissible in \sellU. \end{proposition}

\subsection{Process Reduction,  Cut-Elimination and Identity Expansion}\label{sec:cut}
Consider the  cut-rule below:
\[\infer[Cut]{\seqtype{\Upsilon,\Delta,\Delta'}{\pinew{x}{(P \parallel Q)}}{T}}{
 \deduce{\seqtype{\Upsilon,\Delta}{P}{\TYPELOC{x}{b}{a}{A}}}{}
 &
 \deduce{\seqtype{\Upsilon,\Delta',\TYPELOC{x}{b}{a}{A}}{Q}{T}}{}
}
\]

 The following theorem 
shows that this rule is admissible and, more interestingly, it shows that proof reductions correspond  to reductions in processes. 

\begin{theorem}[Cut-elimination]\label{th:cut-elim}
The cut rule above is admissible. 
\end{theorem}
\begin{proof}
  Let us  present some of the principal cases.

\begin{center}
\resizebox{.95\textwidth}{!}{
$
\begin{array}{c}
\infer[Cut]{\seqtype{\Upsilon, \Delta,
  \Delta_1,\Delta_2}{\pinew{x}{(\boxp{a}{\piinput{x}{b}{y}{p}} \parallel
  \boxp{b}{\pioutput{x}{a}{w}{(q \parallel q')}})}}{T}}{
\infer[\limp_R]{\seqtype{\Upsilon,\Delta}{\boxp{a}{\piinput{x}{b}{y}{p}}}{\TYPELOC{x}{b}{a}{A \limp B}}}{
 \deduce{\seqtype{\Upsilon^{\succeq a},\Delta^{\succeq a},\TYPELOC{y}{a}{b}{A}}{\boxp{a}{{p}}}{\TYPELOC{x}{b}{a}{B}}}{\Pi}
}
&
\infer[\limp_L]{\seqtype{\Upsilon,\Delta_1,\Delta_2,\TYPELOC{x}{b}{a}{A\limp
    B}}{\boxp{b}{\pioutput{x}{a}{w}{(q \parallel q')}}}{T}}{
 \deduce{\seqtype{\Upsilon,\Delta_1}{\boxp{b}{q}}{\TYPELOC{w}{a}{b}{A}}}{}
 &
 \deduce{\seqtype{\Upsilon,\Delta_2,\TYPELOC{x}{b}{a}{B}}{\boxp{b}{q'}}{T}}{}
}
}
\\\\
\rightsquigarrow
\\\\
\infer[{Cut}]{\seqtype{\Upsilon,\Delta,\Delta_1,\Delta_2}{\pinew{w}{(\boxp{b}{q} \parallel \pinew{x}{(\boxp{a}{p\{w/y\}}
  \parallel \boxp{b}{q'}))}}}{T}}{
 \deduce{\seqtype{\Upsilon,\Delta_1}{\boxp{b}{q}}{\TYPELOC{w}{a}{b}{A}}}{}
 &
 \infer[{Cut}]{\seqtype{\Upsilon,\Delta,\Delta_2,\TYPELOC{w}{a}{b}{A}}{\pinew{x}{(\boxp{a}{p\{w/y\}} \parallel \boxp{b}{q'})}}{T}}{
  \deduce{\seqtype{\Upsilon,\Delta,\TYPELOC{w}{a}{b}{A}}{\boxp{a}{{p\{w/y\}}}}{\TYPELOC{x}{b}{a}{B}}}{\Pi'}
  &
  \deduce{\seqtype{\Upsilon,\Delta_2,\TYPELOC{x}{b}{a}{B}}{\boxp{b}{q'}}{T}}{}
 }
}
\end{array}
$
}\end{center}

Proof $\Pi'$ can be obtained from $\Pi$
by using weakening on $\Upsilon$ (and noticing that $\Delta =
\Delta^{\succeq a}$ since $\Delta^{\succeq a}$ is defined) and a
Substitution Lemma, showing that $\Pi\{w/y\}$ is
  indeed a proof. 
The case for $\otimes$ is analogous. 
Let us  present also the case for migration:

\begin{center}
\resizebox{.9\textwidth}{!}{
$
\begin{array}{c}
\infer[{Cut}]{\seqtype{\Upsilon,\Delta,\Delta'}{\pinew{x}{(\boxp{a}{\pimoveout{x}{b}{w}{c}{p}}
  \parallel
  \boxp{b}{\pimovein{x}{a}{y}{c}{q}})}}{T}}{
 \infer[!_R]{\seqtype{\Upsilon,\Delta}{\boxp{a}{\pimoveout{x}{b}{w}{c}{p}}}{\TYPELOC{x}{b}{a}{\nbang{c}{A}}}}{
  \deduce{\seqtype{\Upsilon^{{\succeq} a},\Delta^{{\succeq} a}}{\boxp{c}{{p}}}{\TYPELOC{w}{c}{c}{{A}}}}{\Pi_p}
 }
 &
 \infer[!_L]{\seqtype{\Upsilon,\Delta',\TYPELOC{x}{b}{a}{\nbang{c}{A}}}{\boxp{b}{\pimovein{x}{a}{y}{c}{q}}}{T}}{
  \deduce{\seqtype{\Upsilon,\Delta',\TYPELOC{y}{c}{c}{A}}{\boxp{c}{{q}}}{T}}{\Pi_q}
 }
}
\\\\
\rightsquigarrow
\\
\infer[{Cut}]{\seqtype{\Upsilon,\Delta,\Delta'}{\pinew{w}{(\boxp{c}{p}
  \parallel \boxp{c}{q\{w/y\}})}}{T}}{
 \deduce{\seqtype{\Upsilon,\Delta}{\boxp{c}{{p}}}{\TYPELOC{w}{c}{c}{{A}}}}{\Pi_p'}
 &
 \deduce{\seqtype{\Upsilon,\Delta',\TYPELOC{w}{c}{c}{A}}{\boxp{c}{{q\{w/y\}}}}{T}}{\Pi_q'}
}
\end{array}
$}\end{center}

where $\Pi_p'$ is obtained from 
${\Pi_p}$ by weakening and 
$\Pi_q'$ from $\Pi_q$ by substitution. 
\end{proof}

\begin{corollary}[Type Preservation]
If $\seqtypeA{\mathcal{A}}{\Gamma}{P}{T}$ and $P \Rew{} Q$
then $\seqtypeA{\mathcal{A}}{\Gamma}{Q}{T}$.
\end{corollary}

\begin{proof}
  By induction on $P \Rew{} Q$. Since $P$ is typed, the {unique case is}
for $P$ of the form $\pinew{x}{(P_1 \parallel P_2)}$ where  $P_1
\parallel P_2 \Rew{r} Q'$, $r \in \{\redrule{Comm}, \redrule{LocOut},
\redrule{Move}, \redrule{LCase}, \redrule{RCase}\}$ and $Q =
\pinew{x}{Q'}$. Therefore, the last step in a typing derivation of
$\seqtypeA{\mathcal{A}}{\Gamma}{P}{T}$ is the $(Cut)$ rule and the
result holds by the Cut Elimination Theorem \ref{th:cut-elim}. 
\end{proof}

Another corollary from Theorem \ref{sec:cut} is a global
  progress property.

\begin{corollary}[Global Progress]
If $\seqtypeA{\mathcal{A}}{\cdot }{P}{\TYPELOC{x}{a}{c}{1}}$ then
 either $P \equiv \boxp{c}{\zero}$ or  exists $Q$ s.t. $P\to Q$.
\end{corollary}
Note that, if $P \not\equiv \boxp{c}{\zero}$ then, by a simple
  inspection of the typing rules (see Fig.~\ref{fig-type-sys},
    App. \ref{app:typesystem}), $(Cut)$ is the only possibility as
  the last rule application holding such a typing.

\subsection{Identity Expansion}\label{sec:id-exp}

The copycat  rule  extends the one in~\cite{DBLP:conf/fossacs/ToninhoCP12} with locations: 
$$\infer[\textit{Id}]{\seqtype{\Gamma, \TYPELOC{x}{a}{b}{A}}{\boxp{a}{\pimed{x}{a}{b}{y}}}{\TYPELOC{y}{b}{a}{A}}}{}.$$

As shown in \cite{DBLP:journals/mscs/CairesPT16}, the identity expansion theorem 
reveals the existence of 
synchronous mediators
of the form $x\leftrightarrow y$, a copycat
process that plays the role of a link between   two channels \cite{DBLP:journals/tcs/Sangiorgi96a}. In our system, the same procedure reveals the existence of mediators between different locations. This means  that such processes implement the communication 
between the endpoints located at different sites. 
Remember that 
$\pimed{x}{a}{b}{y}$ is a process  that links the channel $x$ on location $b$ with the channel $y$ on location $a$. 

\begin{theorem}[Identity expansion]\label{th:id-exp}
For all type $A$ 
and locations $a {\preceq} b$, there exists 
a process $id(x,y)$ and a cut-free derivation of the judgment  $\seqtypeA{\mathcal{A}}{\TYPELOC{x}{a}{b}{A}}{\boxp{a}{id(x,y)}}{\TYPELOC{y}{b}{a}{A}}$
\end{theorem}
\begin{proof}
By induction on the structure of $A$. For $A=\one$, $id(x,y)=0$ 
and for an atomic proposition, 
$id(x,y)= \pimed{x}{a}{b}{y}$. Consider the case $A \limp B$:

\begin{center}
\resizebox{0.8\textwidth}{!}{$
\infer[\limp_R]{\seqtype{\TYPELOC{x}{a}{b}{A \limp B}}{\boxp{a}{\piinput{y}{b}{z}{\pioutput{x}{b}{w}{P}}}}{\TYPELOC{y}{b}{a}{A\limp B}}}{
 \infer[\limp_L]{\seqtype{\TYPELOC{x}{a}{b}{A \limp B}, \TYPELOC{z}{a}{b}{A}}{\boxp{a}{\pioutput{x}{b}{w}{P}}}{\TYPELOC{y}{b}{a}{B}}}{
  \deduce{\seqtype{ \TYPELOC{z}{a}{b}{A}}{\boxp{a}{id(z,w)}}{\TYPELOC{w}{b}{a}{A}}}{}
  &~~
  \deduce{\seqtype{\TYPELOC{x}{a}{b}{B}}{\boxp{a}{id(x,y)}}{\TYPELOC{y}{b}{a}{B}}}{}
 }
}
$}\end{center}
Here,  $P=(id(z,w) \parallel id(x,y))$.  What we observe is that the channels $x$ and $z$ located at $b$ 
are linked, respectively, to the channels 
$y$ and $w$ located at $a$. 
The other cases can be obtained similarly. 
\end{proof}

\noindent\textbf{Applications.} We conclude this section by describing some interesting behaviors that can be 
specified and verified with our system. 

\begin{example}[Storing Values]\label{ex:storing}
Process $P$ from Examples~\ref{ex:st.values}
  and~\ref{ex:st.values-red} is not typable
whenever $c \not\succeq sv$.
 Let $A = 1 \otimes 1$, then a type derivation for $P$ is of the form
 \[
 \small
\infer[Cut]{\seqtype{\Upsilon,\Delta}{\pinew{x}{(
\boxp{sv}{\piinput{x}{c}{y}{\pioutput{x}{c}{z}{\piinput{y}{c}{k}{\zero}}}}
\parallel
\boxp{c}{\pioutput{x}{sv}{y}{(\pioutput{y}{}{k}{\zero} \parallel \piinput{x}{sv}{z}{\zero})}}
)}}{T}}{
 \deduce{\seqtype{\Upsilon,\Delta}{\boxp{sv}{\piinput{x}{c}{y}{\pioutput{x}{c}{z}{\piinput{y}{c}{k}{\zero}}}}}{\TYPELOC{x}{c}{sv}{A
     \limp A}}}{\Pi}
 &
 \deduce{\seqtype{\Upsilon,\TYPELOC{x}{c}{sv}{A \limp A}}{\boxp{c}{\pioutput{x}{sv}{y}{(\pioutput{y}{}{k}{\zero} \parallel \piinput{x}{sv}{z}{\zero})}}}{T}}{\Psi}
}
\]
where $\Pi$ is of the form
\begin{center}
\resizebox{.8\textwidth}{!}{
$
\infer[\limp_R]{\seqtype{\Upsilon,\Delta}{\boxp{sv}{\piinput{x}{c}{y}{\pioutput{x}{c}{z}{\piinput{y}{c}{k}{\zero}}}}}{\TYPELOC{x}{c}{sv}{A
        \limp
        A}}}{\infer[\otimes_R]{\seqtype{\Upsilon^{\succeq sv},\Delta^{\succeq sv},\TYPELOC{y}{sv}{c}{A}}{\boxp{sv}{\pioutput{x}{c}{z}{\piinput{y}{c}{k}{\zero}}}}{\TYPELOC{x}{c}{sv}{A}}}{\seqtype{\Upsilon^{\succeq sv},\Delta^{\succeq
        sv}}{\boxp{sv}{\zero}}{\TYPELOC{z}{c}{sv}{1}} &
    \infer[\otimes_L]{\seqtype{\Upsilon^{\succeq sv}, \TYPELOC{y}{sv}{c}{A}}{\boxp{sv}{\piinput{y}{c}{k}{\zero}}}{\TYPELOC{x}{c}{sv}{1}}}{
    \infer[1_L]{\seqtype{\Upsilon^{\succeq sv},
        \TYPELOC{y}{sv}{c}{1},\TYPELOC{k}{sv}{c}{1}}{\boxp{sv}{\zero}}{\TYPELOC{x}{c}{sv}{1}}}{\infer[1_L]{\seqtype{\Upsilon^{\succeq
            sv},
          \TYPELOC{y}{sv}{c}{1}}{\boxp{sv}{\zero}}{\TYPELOC{x}{c}{sv}{1}}}{\infer[1_R]{\seqtype{\Upsilon^{\succeq
              sv}}{\boxp{sv}{\zero}}{\TYPELOC{x}{c}{sv}{1}}}{}}}}}}
$}\end{center}              
Note that the application of $\otimes_R$ is valid only if
$c \succeq sv$. Intuitively, the permission of a value sent by the
  client to be stored and ``used later''  by the server --in this
  example, the channel name $y$ used after the output
  $x \langle z \rangle$-- needs to be explicit in the subexponential signature.

In general, communicated values are either atomic data, of
  type $1$ and its consumption represented by the application of rule $1_L$; or data
  to be used in further interactions. In the first case, the
  use of $1_L$ does not impose any constraint on either location.
  In the second case, the subexponential signature must indicate that such
  further interactions are allowed, as illustrated above. 

\end{example}

\begin{example}\label{ex:quant}
We present an adaptation of the example introduced in
  \cite{DBLP:conf/concur/CairesPPT19}, representing a client-server communication
  defining a buying protocol where sensitive information, such as a
  credit card number, is communicated only after an agreed secure
  channel is established. In other words, we want to type
  $Sys := \pinew{x}{(\boxp{ws}{W} \parallel \boxp{c}{C})}$ such that
{\small
\[\infer[Cut]{\seqtype{\Upsilon}{\pinew{x}{(\boxp{ws}{W} \parallel
        \boxp{c}{C})}}{\TYPELOC{z}{ws}{c}{1}}}{ \deduce{\seqtype{\Upsilon}{\boxp{ws}{W}}{\TYPELOC{x}{c}{ws}{WT}}}{\Pi} & \deduce{\seqtype{\Upsilon,\TYPELOC{x}{c}{ws}{WT}}{\boxp{c}{C}}{\TYPELOC{z}{ws}{c}{1}}}{\Psi}}\]}
for (pre-)types
$
\begin{array}{lll@{\hspace{10mm}}lll}
WT & := & addC \limp  (\existsLoc
          _{\Typeloc{\alpha}{\secure}}.\nbang{\alpha}Pay) \with 1 &
Pay & := & CCNum \limp (\forallLoc
           _{\Typeloc{\beta}{\alpha}}.\nbang{\beta}{Rcpt}) \oplus 1
\end{array}
$

Here,  $addC$ denotes an item information added by the client in the online
  chart and $CCNum$ denotes the credit card number sent by the
  client. The intended meaning of  \ $WT$ is that, once the item
  information is processed, the client can choose to either close the
  session or go to payment. The latter depends on the migration of both, client and
  server, to a secure location ($l_s$ below)  chosen by the server but agreed by both
sides of the communication. Note that the credit card number would be sent
only if the (channel) migration succeeds. 
In the end, after receiving $CCNum$,
the processes migrate to yet another location ($l_b$) to complete the payment. 
Hence, we  consider the following set of locations:  $I = \{c
\prec ws, \secure \prec ws, \secure \prec c, l_s \prec \secure, l_b \prec l_s \}$
and

$
\begin{array}{lll@{\hspace{10mm}}lll}
C & := &
         \pioutput{x}{ws}{cat}{\piinl{x}{ws}{\piwin{x}{ws}{\alpha}{C_1}}}
  &
 W & :=
  &\piinput{x}{\cdaniel{X}{c}}{cat}{\picase{x}{Y}{\pioutputD{x}{\cdaniel{\secure}{c}}{l_s}{W_1}}{\zero}}\\
  
  C_1 & := &
             \pimovein{x}{ws}{x'}{\alpha}{\pioutput{x'}{\alpha}{n}{C_2}}
  &
  W_1 & := &
             \pimoveout{x}{Z}{x'}{l_s}{\piinput{x'}{U}{n}{W_2}}\\
  
 C_2 & := &
            \picase{x'}{\alpha}{\pioutputD{x'}{\alpha}{l_b}{\pimovein{x'}{\alpha}{y}{l_b}{\zero}}}{\zero}  
  &
 W_2 & := & \piinl{x'}{V}{\piwin{x'}{\cdaniel{\bnk}{l_s}}{\beta}{\pimoveout{x'}{T}{y}{\beta}{\zero}}}
\end{array}
$

Here, $\secure$ is used as the type of locations known by
both client and server since $\secure \prec c, \secure \prec ws$.  $l_s$ is the location  of type $\secure$ ($l_s \prec \secure$) to be used by the server.
Moreover, $l_b$ (of type $l_s$)  
is the location chosen by the client to complete the payment. 

Derivation $\Pi$ takes the form 
\begin{center}
$
\infer[\limp_R]{\seqtype{\Upsilon}{\boxp{ws}{W}}{\TYPELOC{x}{c}{ws}{WT}}}{
  \infer[1_L]{\seqtype{\Upsilon^{\succeq ws},\TYPELOC{cat}{ws}{c}{addC}}{\boxp{ws}{W'}}{\TYPELOC{x}{c}{ws}{WT'}}}{\infer[\with_R]{\seqtype{\Upsilon^{\succeq ws}}{\boxp{ws}{W'}}{\TYPELOC{x}{c}{ws}{WT'}}}{
  \infer[\existsLoc_R]{\seqtype{\Upsilon^{\succeq
      ws}}{\boxp{ws}{\piwout{x}{\cdaniel{\secure}{c}}{l_s}{W_1}}}{\TYPELOC{x}{c}{ws}{\existsLoc_{\Typeloc{\alpha}{\secure}}.\nbang{\alpha}Pay}}}{\infer[!_R]{\seqtype{\Upsilon^{\succeq
      ws}}{\boxp{ws}{W_1}}{\TYPELOC{x}{c}{ws}{\nbang{l_s}Pay\{l_s/\alpha\}}}}{\deduce{\seqtype{\Upsilon^{\succeq
      ws}}{\boxp{l_s}{W'_1}}{\TYPELOC{x'}{l_s}{l_s}{Pay\{l_s/\alpha\}}}}{\Pi'}}}
  & \infer[1_R]{\seqtype{\Upsilon^{\succeq
        ws}}{\boxp{ws}{\zero}}{\TYPELOC{x}{c}{ws}{1}}}{}}}}
$
\end{center}
Note that $\TYPELOC{cat}{ws}{c}{addC}$ needs to be consumed before
  the choice (rule $\with_R$).
  
Derivation $\Pi'$ is as follows: 
\begin{center}
$
\infer[\limp_R]{\seqtype{\Upsilon^{\succeq
      ws}}{\boxp{l_s}{W'_1}}{\TYPELOC{x'}{l_s}{l_s}{Pay\{l_s/\alpha\}}}}{\infer[\oplus_{R1}]{\seqtype{(\Upsilon^{\succeq
        ws})^{\succeq l_s},\TYPELOC{n}{l_s}{l_s}{CCNum}}{\boxp{l_s}{W_2}}{\TYPELOC{x'}{l_s}{l_s}{Pay'\{l_s/\alpha\}}}}{\infer[\forallLoc_R]{\seqtype{(\Upsilon^{\succeq
        ws})^{\succeq l_s},\TYPELOC{n}{l_s}{l_s}{CCNum}}{\boxp{l_s}{W'_2}}{\TYPELOC{x'}{l_s}{l_s}{\forallLoc_{\Typeloc{\beta}{{l_s}}}.\nbang{\beta}{Rcpt}}}}{\infer[!_R]{\seqtypeA{\beta
           \preceq l_s}{(\Upsilon^{\succeq
        ws})^{\succeq
          l_s},\TYPELOC{n}{l_s}{l_s}{CCNum}}{\boxp{l_s}{\pimoveout{x'}{T}{y}{\beta}{\zero}}}{\TYPELOC{x'}{l_s}{l_s}{\nbang{\beta}{Rcpt}}}}{\infer[1_L]{\seqtypeA{\beta
        \preceq{l_s}}{(\Upsilon^{\succeq
        ws})^{\succeq
          l_s},\TYPELOC{n}{l_s}{l_s}{CCNum}}{\boxp{\beta}{\zero}}{\TYPELOC{y}{\beta}{\beta}{Rcpt}}}{\infer[1_R]{\seqtypeA{\beta
        \preceq l_s}{(\Upsilon^{\succeq
        ws})^{\succeq
          l_s}}{\boxp{\beta}{\zero}}{\TYPELOC{y}{\beta}{\beta}{Rcpt}}}{}}}}}}
$
\end{center}
\end{example}
\label{sec:examples}

\section{Related Work and Concluding Remarks}\label{sec:discussion}
This paper originates from the work of 
Caires et al. \cite{DBLP:conf/concur/CairesPPT19} where 
Hybrid Linear Logic (HyLL) was proposed as the foundation for domain aware session types. 
Let us compare the two systems from the point of view of the underlying logic 
and the resulting process calculi.

\paragraph{Logical foundation.} 
Hybrid Linear Logic (HyLL) is a conservative extension of intuitionistic LL where the truth judgments are labeled by worlds as in $A@w$, read as ``A holds at world w''. 
The reachability relation on worlds,  $\prec: W \times W$,  determines
when a world $w$ is accessible from  $v$ ($v \prec w)$.

Formulas in HyLL are built from the connectives of intuitionistic linear logic  plus the following ones:

$
\qquad F ::= ... \mid (A ~\at\ w) \mid \down w. A \mid \forall w. A \mid \exists w. A
$

where the world $w$ is bounded in the formulas
$\down w. A$, $\forall w. A$ and  $\exists w. A$. 

The sequent calculus
for HyLL \cite{DBLP:conf/types/DespeyrouxC13} uses dyadic sequents of
the form $\Gamma; \Delta \vdash C ~@~ w$ where
the unbounded ($\Gamma$) and linear ($\Delta$) contexts contain judgments of the form $A ~@~ w$. 
Let us introduce some of the rules as originally proposed in  \cite{DBLP:conf/types/DespeyrouxC13}. The initial rule enforces that atomic propositions 
must be located at the same world. The 
other rules are similar to those of intuitionistic LL: \\

$\small
\infer[{I}]{\Gamma ; \varp ~@~ w \vdash \varp ~@~ w}{}
\qquad
\infer[{ \otimes R}] {\Gamma ; \Delta, \Delta' \vdash A \otimes B ~@~ w}{\Gamma ; \Delta \vdash A ~@~ w  \quad  \Gamma ; \Delta' \vdash B ~@~ w}    
\qquad \infer[{\otimes L}]{\Gamma ; \Delta, A \otimes B ~@~ u \vdash C ~@~ w}{\Gamma ; \Delta, A ~@~ u, B ~@~ u \vdash C ~@~ w}    
$

The interesting part is the rules for the hybrid connectives: \\

\noindent\resizebox{\textwidth}{!}{
$
\begin{array}{c}
  \infer[{\at~ R}] {\Gamma ; \Delta \vdash (A ~\at~ u) ~@~ w}{\Gamma ; \Delta \vdash A ~@~ u} 
 \qquad
  \infer[{\at~ L}]{\Gamma ; \Delta, (A ~\at~ u) ~@~ w \vdash C ~@~ v} {\Gamma ; \Delta, A ~@~ u \vdash C ~@~ v}
\qquad
  \infer[{\downarrow R}]{\Gamma ; \Delta \vdash \downarrow u. A ~@~ w}{\Gamma ; \Delta \vdash A [w / u] ~@~ w}    
  \qquad
  \infer[{\downarrow L}]{\Gamma ; \Delta, \downarrow u. A ~@~ w \vdash C ~@~ v}{\Gamma ; \Delta, A [w / u] ~@~ w \vdash C ~@~ v}
  \\\\
  \infer[{\forall R}]{\Gamma ; \Delta \vdash \forall \alpha.~ A ~@~ w}{\Gamma ; \Delta \vdash A ~@~ w}   
  \qquad
  \infer[{\forall L}]{\Gamma ; \Delta, \forall \alpha.~ A ~@~ u \vdash C ~@~ w}{\Gamma ; \Delta, A [\tau/ \alpha] ~@~ u \vdash C ~@~ w}
  \end{array}
$
}

The formula $(A ~\at~ u)$ is a \emph{mobile} proposition:
it carries  the world at which it is true.
The \emph{localization} ($\downarrow$)  binds a name for the current world where
the proposition is true.
Hence, the formula  $(\downarrow u.A)~@~w$ 
\emph{fixes} the world $w$ in $A$  by substituting the occurrences of $u$ with $w$ in $A$. 
The rules for the universal quantifier are as expected. The rules for existential quantification on worlds  are dual and omitted. 

Note that the rules for the hybrid connectives  are not constrained by the  context $\Gamma ; \Delta$. More precisely, the rule $\at_R$
is able to move the formula $A$ from $w$ to $u$ without checking if $u$ is reachable from the worlds in the context $\Gamma;\Delta$ (recall that \sellU
checks such condition in the promotion rule).
Similarly for the quantifiers:  $\forall_L$ substitutes $\alpha$ with \emph{any} $\tau$ (related or not with $u$, $w$, etc). 
This  observation was fundamental 
in the expressiveness study in   \cite{DBLP:journals/mscs/ChaudhuriDOP19}, where it is shown that the inferences rules of HyLL can be adequately encoded in (vanilla) linear logic. Moreover, a second encoding of HyLL into \sellU\  showed  that  a flat subexponential structure  is enough to encode any world structure in HyLL. 
This is partially explained by the fact that 
the reachability relation  on worlds
does not play any role in HyLL's inference rules. 

With the aim of controlling migration of formulas between worlds,  \cite{DBLP:conf/concur/CairesPPT19} adds some side conditions to the system for HyLL. 
For instance, the rules for the connective $\at$ are extended as follows: 

\begin{itemize}
 \item $\at R$. Two side conditions are added: $w \prec  u$ and $u \prec^* \Delta$. The second conditions uses the reflexive and transitive closure of $\prec$ and tells us that for any formula $G~@~x$ in the context $\Delta$, it must be the case that $w \prec^* x$. 
 \item $\at L$. No side condition is added but, bottom up,  the rule extends $\prec$
 with the new entry $w \prec u$. 
\end{itemize}

These side conditions (and others imposed in the rest of the rules) guarantee that the system produces only sequents of the form 
$\Gamma;\Delta \vdash A~@~w$ where
all the worlds in $\Delta$ are reachable from  $w$. 
Note also that, due to the new restriction on $\at R$, the movement of formulas on the right side of the sequent can take place only to an \emph{immediate successor} world ($w \prec u$). 

\paragraph{Processes.}
The 
language of processes considered in 
\cite{DBLP:conf/concur/CairesPPT19} does not make explicit the location of an agent. In fact,  what we call a pre-process here, is a process in \cite{DBLP:conf/concur/CairesPPT19}. 
This implies that,   when a process migrates, the information about its location cannot be determined by simply inspecting the terms/processes. 
(We are not aware whether such information can be recovered somehow in the system in \cite{DBLP:conf/concur/CairesPPT19}).
For instance, consider the following derivations in both systems: 
\[
\begin{array}{lllccccc}
\mbox{System in \cite{DBLP:conf/concur/CairesPPT19}:} & & 
\pimoveoutHy{x}{y}{w}{p}& \parallel &
\pimoveinHy{x}{z}{w}{q} &\redi{}&
p \parallel q\{y/z\}\\
\mbox{Our system:} & &
\boxp{a}{\pimoveout{x}{b}{y}{w}{p}} & \parallel &
\boxp{b}{\pimovein{x}{a}{z}{w}{q}} &\redi{}&
\boxp{w}{p \parallel q\{y/z\}}
\end{array}
\]

\paragraph{Type System.}
Modifying the rules of the system for HyLL allows for controlling
the context when formulas are moved between worlds. However, extending the signature 
in $\at L$
is problematic.
The cut-elimination procedure 
reveals that such extension is not necessary but the identity-expansion theorem requires it. 
In Appendix~\ref{app-prop-ex}, we elaborate more about this potential problem in the interpretation
proposed in \cite{DBLP:conf/concur/CairesPPT19} of locations as worlds in HyLL.

Our type system is founded on a 
logic where the promotion rule controls the resources that can be used during a
proof. Moreover, consuming resources from the context, i.e.,
performing dereliction, is not subject to any restriction. The act of
 communicating locations in our proposal seems also 
to be more natural: a
  communicated location is accessible from  both endpoints of the communication (notation $a\sqcap b$). 
 As shown in Section \ref{sec:examples}, our type system allows for controlling the information that can be stored in a ``non-trusted'' (i.e., non-related)  site and also, to correctly reproduce the example in \cite{DBLP:conf/concur/CairesPPT19}. Finally, as stated before, the syntax of processes reflects more precisely the intention of the modeler. 
 
 As shown in different works (see e.g.,
 \cite{DBLP:journals/mscs/ChaudhuriDOP19,DBLP:journals/tcs/NigamOP17,DBLP:journals/tcs/OlartePN15,DBLP:journals/tplp/PimentelON14}),
the subexponentials in \sell\ can be interpreted in
 different ways. For instance, they can be used to denote time-units,
 spaces of computation, the epistemic state of agents, preferences,
 costs, etc. It would be interesting to see if the interpretation
 proposed here can be seen through the lenses of those modalities. For
 instance, instead of interpreting
 $\boxp{a}{p}$ as ``$p$ {located at} $a$'', we may interpret $a$ as a deadline for the execution of $p$ (or ``$p$ can use a resource in the following $a$ time-units''). This may lead to  a system with 
 session duration/expiration  with strong foundation in logical methods.

\bibliographystyle{eptcs}

\appendix

\section{Locations as HyLL worlds} \label{app-prop-ex} 

Let  us consider the 
following sequent
$
\Gamma ; x:(A~\at\, \secretp)@w \vdash C :: z:T@c
$ in 
\cite[Section 3]{DBLP:conf/concur/CairesPPT19}, 
where $w$ means \emph{web store} and $c$ stands for
\emph{client}. 
 It is assumed in \cite{DBLP:conf/concur/CairesPPT19} that the world signature, at this point, only contains the entry 
$c \prec w$. Since $c \not\prec^* \secretp$,
first the server has to migrate to the secure location $\secretp$
(using $\at L$) and, only after that, the client is able to use a service from that location. According to op. cite, ``\emph{this ensures, e.g., that a client cannot exploit the payment platform of the web store by accessing the trusted domain  in unforeseen ways}''. 
This behavior is not reflected by the cut-elimination procedure as the following proposition shows:
\begin{proposition}
Let $C=\pimoveinHy{x}{y}{\secretp}{\zero}$, $W=\pimoveoutHy{x}{y}{\secretp}{\zero}$
and $S = \pinew{x}{(W \parallel C)}$. 
Hence, the sequent $\cdot;\cdot \vdash S :: - : \one @c$ is provable iff
$c \prec w$, $c \prec \secretp$ and $w \prec \secretp$.
\end{proposition}
\begin{proof}
We shall prove something more general. Since the only way of proceeding in the typing of $S$ is to use the rule cut (as proposed in \cite{DBLP:conf/concur/CairesPPT19}), we shall leave some meta-variables ($V_i$) in such a derivation to find  the restrictions needed to complete the proof. Also, we shall use $\Omega$ to denote the context proving 
statements of the form $v \prec w$: 
\[
\infer[Cut]{\Omega ; \cdot;\cdot \vdash \pinew{x}{(W \parallel C) :: -:\one @ V_1}}{
  \infer[\at R]{\Omega ; \cdot;\cdot \vdash W :: x:(T_w'~ \at \secretp)@V_2}{
   \deduce{\Omega ; \cdot;\cdot \vdash W' :: y: T_w' @ \secretp}{}
  }
  &
  \infer[\at L]{\Omega ; \cdot ; x : (T_w' \at \secretp) @ V_2 \vdash C :: - :1 @ V_1}{
   \deduce{\Omega, V_2 \prec \secretp ; \cdot; y : T'_w @ \secretp \vdash C' :: -:1 @ V_1}{}
  }
}
\]
According to the rules in \cite{DBLP:conf/concur/CairesPPT19}, the following must hold:

\[
\begin{array}{llr}
\Omega \vdash V_1 \prec^* V_2 & \qquad & \mbox{due to the cut-rule}\\
\Omega  \vdash V_2 \prec \secretp & \qquad & \mbox{due to $\at R$}\\
\end{array}
\]

Hence, the addition of ``$V_2 \prec \secretp$'' in $\at L$ is irrelevant (since it must be already known in $\Omega$). 

Without that extension, the initial-expansion theorem does not work (omitting the terms/processes), 

\[
\infer[\at L]{\Omega ; \cdot ; (A \at V_2)@V_1 \vdash (A \at V_2)@V_1}{
 \infer[\at R]{\Omega,V_1\prec V_2 ; \cdot ; A@V_2 \vdash (A \at V_2)@V_1}{
  \deduce{\Omega,V_1\prec V_2 ; \cdot ; A@V_2 \vdash A@V_2}{}
 }
}
\]
since $\at R$ requires that $V_1 \prec V_2$. 
\end{proof}

This proposition is  telling us that: (1) there is no need to extend the world signature in $\at L$ since, the ``new'' relation $w \prec \secretp$ must be already known (otherwise, the process $S$ is not typable).
And (2), 
adding  $w \prec \secretp$ is not enough
to guarantee that the   client can move to $\secretp$: from $c\prec w$ and $w\prec \secretp$ we  conclude $c \prec^* \secretp$ but, the side condition in $\at R$
requires $c \prec \secretp$. This explains 
the signature $\Omega$ ($c \prec w$, $c \prec \secretp$ and $w \prec \secretp$) in the above proposition.

According to 
\cite{DBLP:conf/concur/CairesPPT19}, 
the extension of the signature in $\at L$ was added in order to prove the identity expansion theorem (as shown above). However, the price to be paid is that the cut-elimination procedure does not reflect any more the intended meaning in the example above.
\section{Type system}\label{app:typesystem}
In Fig. \ref{fig-type-sys} we show the complete typing system. 
\begin{figure}[h]
\resizebox{.96\textwidth}{!}
{$\begin{array}{c}
\infer[\textit{Id}]{\seqtype{\Gamma, \TYPELOC{x}{a}{b}{A}}{\boxp{a}{\pimed{x}{a}{b}{y}}}{\TYPELOC{y}{b}{a}{A}}}{}
\quad
      \infer[\one_L]{\seqtype{\Gamma, \TYPELOC{x}{b}{a}{\one}}{P}{T}}{\seqtype{\Gamma}{P}{T}}
\qquad
      \infer[\one_R]{\seqtype{\Upsilon}{\boxp{b}{\zero}}{\TYPELOC{x}{a}{b}{\one}}}{}
      \\\\
\infer[\tensor_L]{\seqtype{\Gamma, \TYPELOC{x}{b}{a}{A \otimes B}}{\boxp{b}{\piinput{x}{a}{y}{p}}}{T}}{
\seqtype{\Gamma, {\TYPELOC{x}{b}{a}{B}}, {\TYPELOC{y}{b}{a}{A}} }{\boxp{b}{p}}{T}
}
\qquad
\infer[\tensor_R]{\seqtype{\Upsilon, \Delta,\Delta'}{\boxp{b}{\pioutput{x}{a}{y}{(p \parallel q)}}}{\TYPELOC{x}{a}{b}{A\otimes B}}}
{
 \deduce{\seqtype{\Upsilon^{\succeq b},\Delta^{\succeq b}}{\boxp{b}{p}}{{\TYPELOC{y}{a}{b}{A}}}}{}
 &
  \deduce{\seqtype{\Upsilon^{\succeq b}, \Delta'^{\succeq b}}{\boxp{b}{q}}{{\TYPELOC{x}{a}{b}{B}}}}{}
}
      \\\\
\infer[\limp_L]{\seqtype{\Upsilon,\Delta,\Delta',\TYPELOC{x}{b}{a}{A\limp B}}{\boxp{b}{\pioutput{x}{a}{y}{(p \parallel q)}}}{T}}{
 \deduce{\seqtype{\Upsilon,\Delta}{\boxp{b}{p}}{{\TYPELOC{y}{a}{b}{A}}}}{}
 &
  \deduce{\seqtype{\Upsilon,\Delta'{,\TYPELOC{x}{b}{a}{B}}}{\boxp{b}{q}}{{T}}}{}
}
 \qquad
  \infer[\limp_R]{\seqtype{\Gamma}{\boxp{b}{\piinput{x}{a}{y}{p}}}{\TYPELOC{x}{a}{b}{A\limp B}}}{
  \seqtype{\Gamma^{\succeq b},{\TYPELOC{y}{b}{a}{A}}}{\boxp{b}{p}}{{\TYPELOC{x}{a}{b}{B}}}
  }
      \\\\
   \infer[\oplus_L]{\seqtype{\Gamma, \TYPELOC{x}{b}{a}{A \oplus B}}{\boxp{b}{\picase{x}{a}{p}{q}}}{T}}{
\seqtype{\Gamma, {\TYPELOC{x}{b}{a}{A}} }{\boxp{b}{p}}{T}
&
\seqtype{\Gamma, {\TYPELOC{x}{b}{a}{B}} }{\boxp{b}{q}}{T}
}
\\\\
\infer[\oplus_{R1}]{\seqtype{\Gamma}{\boxp{b}{\piinl{x}{a}{p}}}{\TYPELOC{x}{a}{b}{A\oplus B}}}
{
 \deduce{\seqtype{\Gamma^{\succeq b}}{\boxp{b}{p}}{{\TYPELOC{x}{a}{b}{A}}}}{}
}
\qquad
\infer[\oplus_{R2}]{\seqtype{\Gamma}{\boxp{b}{\piinr{x}{a}{p}}}{\TYPELOC{x}{a}{b}{A\oplus B}}}
{
 \deduce{\seqtype{\Gamma^{\succeq b}}{\boxp{b}{p}}{{\TYPELOC{x}{a}{b}{B}}}}{}
}
\\\\
\infer[\with_{L1}]{\seqtype{\Gamma,\TYPELOC{x}{b}{a}{A\with B}}{\boxp{b}{\piinl{x}{a}{p}}}{T}}{
\seqtype{\Gamma,\TYPELOC{x}{b}{a}{A}}{\boxp{b}{p}}{T}
}
\qquad
\infer[\with_{L2}]{\seqtype{\Gamma,\TYPELOC{x}{b}{a}{A\with B}}{\boxp{b}{\piinr{x}{a}{p}}}{T}}{
\seqtype{\Gamma,\TYPELOC{x}{b}{a}{B}}{\boxp{b}{p}}{T}
}
\\\\
\infer[\with_R]{\seqtype{\Gamma}{\boxp{b}{\picase{x}{a}{p}{q}}}{\TYPELOC{x}{a}{b}{A \with B}}}{
\seqtype{\Gamma^{\succeq b}}{\boxp{b}{p}}{\TYPELOC{x}{a}{b}{A}}
&
\seqtype{\Gamma^{\succeq b}}{\boxp{b}{q}}{\TYPELOC{x}{a}{b}{B}}
}
  \\\\
\infer[\bang_L]{\seqtype{\Gamma,\TYPELOC{x}{b}{a}{\nbang{c}{A}}}{\boxp{b}{\pimovein{x}{a}{y}{c}{p}}}{T}}{
\seqtype{\Gamma,\TYPELOC{y}{c}{c}{A}}{\boxp{c}{p}}{T}
}
\qquad
\infer[\bang_R]{\seqtype{\Gamma}{\boxp{b}{\pimoveout{x}{a}{y}{c}{p}}}{\TYPELOC{x}{a}{b}{\nbang{c}A}}}{
\seqtype{\Gamma^{\succeq b}}{\boxp{c}{p}}{\TYPELOC{y}{c}{c}{A}}
}
                                                                                    
      \\\\
      \infer[\existsLoc_L]{\seqtypeA{\mathcal{A}}{\Gamma,\TYPELOC{x}{b}{a}{\existsLoc_{\Typeloc{\alpha}{a\sqcap b}}.A}}{\boxp{b}{\piwin{x}{a}{\alpha}{p}}}{T}}{
\seqtypeA{\mathcal{A},\alpha\preceq a\sqcap b}{\Gamma,{\TYPELOC{x}{b}{a}{A}}}{\boxp{b}{p}}{T}
}
\qquad
\infer[\existsLoc_R]{\seqtype{\Gamma}{\boxp{b}{\piwout{x}{a}{{l}}{p}}}{\TYPELOC{x}{a}{b}{\existsLoc _{\Typeloc{\alpha}{a\sqcap b}}.A}}}{
\seqtype{\Gamma^{\succeq b}}{\boxp{b}{p}}{{\TYPELOC{x}{a}{b}{A\{l/\alpha\}}}}
}
      \\\\
\infer[\forallLoc_L]{\seqtype{\Gamma,\TYPELOC{x}{b}{a}{\forallLoc_{\Typeloc{\alpha}{a\sqcap b}}.A}}{\boxp{b}{\piwout{x}{a}{l}{p}}}{T}}{
\seqtype{\Gamma,{\TYPELOC{x}{b}{a}{A\{l/\alpha\}}}}{\boxp{b}{p}}{T}
}
\qquad
\infer[\forallLoc_R]{\seqtypeA{\mathcal{A}}{\Gamma}{\boxp{b}{\piwin{x}{a}{\alpha}{p}}}{\TYPELOC{x}{a}{b}{\forallLoc _{\Typeloc{\alpha}{a\sqcap b}}.A}}}{
\seqtypeA{\mathcal{A},\alpha\preceq a\sqcap b}{\Gamma^{\succeq b}}{\boxp{b}{p}}{{\TYPELOC{x}{a}{b}{A}}}
}
      \\\\
\infer[copy]{\seqtype{\Gamma,\TYPELOC{x}{b}{\unb{a}}{\nbang{a}A}}{\boxp{b}{\pioutput{x}{a}{y}{p}}}{T}}{
\seqtype{\Gamma,\TYPELOC{x}{b}{\unb{a}}{\nbang{a}A}, \TYPELOC{y}{b}{a}{A}}{\boxp{b}{{p}}}{T}
}
\qquad
\infer[!^+_R]{\seqtype{\Upsilon}{\boxp{b}{!\piinput{x}{a}{y}{p}}}{\TYPELOC{x}{a}{\unb{b}}{\nbang{b}A}}}{
\seqtype{\Upsilon^{\succeq \unb{b}}}{\boxp{b}{{p}}}{\TYPELOC{y}{a}{b}{A}}
}
      \\\\
      \infer[Cut]{\seqtype{\Upsilon,\Delta,\Delta'}{\pinew{x}{(P \parallel Q)}}{T}}{
 \deduce{\seqtype{\Upsilon,\Delta}{P}{\TYPELOC{x}{b}{a}{A}}}{}
 &
 \deduce{\seqtype{\Upsilon,\Delta',\TYPELOC{x}{b}{a}{A}}{Q}{T}}{}
}
  \end{array}$}
\caption{Typing System. Since the context $\mathcal{A}$ is only modified in rules $\existsLoc_L$ and $\forallLoc_R$ , we omit it in the other rules. \label{fig-type-sys}}
\end{figure}
\end{document}